\documentclass[10pt,sigconf]{acmart} % O tamanho da fonte tem que ser 10pt https://sigmod2020.org/calls_papers_sigmod_research.shtml#formatting

\usepackage[utf8]{inputenc}

\usepackage{subcaption}
\usepackage{framed}
\usepackage{amsfonts,amssymb,amsmath,mathrsfs}%,oldlfont,newlfont,wasysym} %,amsthm}

\usepackage[mathcal]{euscript}

\usepackage{listings}

\usepackage{booktabs} % For formal tables

\usepackage{tikz, tabularx} % for drawing
\usetikzlibrary{patterns}
\usepackage{pgfplots}
\usetikzlibrary{positioning, arrows}  % use: above left of, etc
\usepackage{hyperref} % for links

\usepackage[ruled, linesnumbered, noend, vlined]{algorithm2e} % For algorithms (`noline` option hides vertical lines)
\SetKwBlock{Globals}{globals}{end}
\SetKw{Let}{let}
\SetKw{inlineSwitch}{switch}

\newcommand{\derives}{\Rightarrow^*}
\newcommand{\derivesk}[1]{\Rightarrow^#1}

\def\bigO{\hbox{$\mathcal {O}$}}
\newcommand{\tuple}[1]{( #1 )}
\newcommand{\set}[1]{\{ #1 \}}
\renewcommand{\emptyset}{\set{~}}
\newcommand{\Rule}[2]{#1 \rightarrow #2}
\newcommand{\Item}[2]{\ensuremath{{[}\ \Rule{#1}{#2}\ {]}}}
\newcommand{\Gnk}{\ensuremath{\mathcal{G}}}

\newcommand{\R}[2]{\ensuremath{\complement_{G,D}(\mathit{#1},\mathit{#2}})}

\newcommand{\rules}{P}
\newcommand{\nonterminals}{N}
\newcommand{\terminals}{\ensuremath{\Sigma}}

\newcommand{\pow}[1]{\ensuremath{\mathscr{P}(#1)}}
\newcommand{\m}{^{\bullet}}
\renewcommand{\u}{^{\circ}}

\newcommand{\traces}{traces}
\newcommand{\paths}{paths}

\newcommand\ucirc{\mathbin{\ooalign{$\cup$\cr%
   \hfil\raise0.42ex\hbox{$\scriptscriptstyle\Join$}\hfil\cr}}}

\newcommand{\munion}{\ucirc{}}

\newcommand{\st}{\mathbf{~s.t.~}}

\newcommand{\I}{\ensuremath{I}}

\newtheorem{grammar}{Grammar}

\newcommand{\defref}[2]{Definition~\ref{def:#1}.\ref{def:#1-#2}}

\usepackage{newfloat}
\usepackage{caption}
\DeclareFloatingEnvironment[fileext=frm,placement={!ht},name=Program]{program}
\newcounter{numberInTrivlist}

\newcommand{\ourremark}[3]{}%\noindent{{\footnotesize\sffamily\textcolor{#2}{#3\textsuperscript{#1}}}}}

\hypersetup{draft}

\AtBeginDocument{%
  \providecommand\BibTeX{{%
    \normalfont B\kern-0.5em{\scshape i\kern-0.25em b}\kern-0.8em\TeX}}}

\setcopyright{acmcopyright}
\copyrightyear{2018}
\acmYear{2018}
\acmDOI{10.1145/1122445.1122456}

\acmConference[SIGMOD '20]{SIGMOD '20: ACM SIGMOD/PODS International Conference on Management of Data}{June 14--19, 2020}{Portland, OR}
\acmISBN{978-1-4503-9999-9/18/06}

\begin{document}

%\title{CFPQ Evaluation with (Shared?) Generalized Items}
%\title{Context-Free Path Queries over RDF Databases}
\title{An Algorithm for Context-Free Path Queries over Graph Databases}

\author{Ciro M. Medeiros}
\email{cirommed@ppgsc.ufrn.br}
\orcid{https://orcid.org/0000-0002-3928-5053}
\author{Martin A. Musicante}
\author{Umberto S. Costa}
\email{ { mam, umberto }@dimap.ufrn.br}
\affiliation{%
 \institution{Federal University of Rio Grande do Norte}
 \streetaddress{Campus Universitário Lagoa Nova}
 \city{Natal}
 \state{RN}
 \country{Brazil}
 \postcode{1524}
}

%\renewcommand{\shortauthors}{C. Medeiros et al.}

\begin{abstract}
RDF (Resource Description Framework) is a standard language to represent graph databases.
Query languages for RDF databases usually include primitives to support path queries, linking pairs of vertices of the graph that are connected by a path of labels belonging to a given language.
Languages such as SPARQL include support for paths defined by regular languages (by means of Regular Expressions).
A context-free path query is a path query whose language can be defined by a context-free grammar. 
Context-free path queries can be used to implement queries such as the ``same generation queries'',  that are not expressible by Regular Expressions. 
In this paper, we present a novel algorithm for context-free path query processing.
We prove the correctness of our approach and show its run-time and memory complexity.
We show the viability of our approach by means of a prototype implemented in Go.
We run our prototype using the same cases of study as proposed in recent works, comparing our results with another, recently published algorithm.
The experiments include both synthetic and real RDF databases.
Our algorithm can be seen as a step forward, towards the implementation of more expressive query languages.

%For the real case experiment, our algorithm improves time in one order of magnitude, when compared to the previous algorithm.
\end{abstract}




\maketitle

\section{Introduction}
Processing a Path Query over a Graph Database consists of looking for pairs of vertices such that they are connected by a specified path inside the graph.
The labels of the edges in a path form a string and, as such, they can be specified by using grammars or other formal tools.
Regular Expressions have been widely used to define path queries.
As regular languages belong to the most restricted class of formal languages, the expressivity of such queries is somehow limited.
Recent studies have developed algorithms for supporting the use of context-free grammars in path queries in order to improve their expressiveness.

% Recent efforts on publishing content comprehensible both for humans and machines are being driven by Linked Data\footnote{\url{https://www.w3.org/wiki/LinkedData}}.
% The term refers to a set of best practices for publishing data on the Web.
% Its objective is to provide standards for representing structured data and making it publicly available.
% To do so, they adopt IRIs to globally identify resources and allow the stating of relationships between them in a way similar to what happens in web pages.
% That allows different databases to be related, what results on a global data graph.

RDF (Resource Description Framework) is the Linked Data standard for representing data.
An RDF database consists on a set of triples that can be viewed as a graph.
The standard query language for RDF databases is SPARQL.
The language supports the definition of paths using regular expressions over labels of edges in the graph.
However, some applications require more sophisticated queries, which cannot be defined using regular expressions, but may be described by context-free grammars.

In the last few years,  a number of initiatives were developed to improve the expressiveness of SPARQL and path query languages in general.
most of these initiatives include de definition of algorithms for the evaluation of context-free path queries.
Such algorithms are, in general, based on parsing techniques.
%Each of these algorithms has its strengths and weaknesses concerning to performance, restrictions on the grammar, direction of path recognition, etc.
In this paper we present a new approach that, while it is not based on a specific parsing technique, it uses annotations over grammar items to parse several paths at the same time, keeping track of shared prefixes over these paths.

Our main contributions are:
\begin{itemize}
    \item an algorithm for evaluation of context-free path queries;
    \item an analysis of correctness, as well as time and space complexity for the algorithm;
    \item experimental results that demonstrate its applicability in different scenarios.
\end{itemize}

\section{Grammars, Data Graphs and Queries}
This section briefly presents some basic background that is used in the paper.

%\subsection{Formal Languages Theory}
%\label{sec:formal-languages}
\begin{definition}[Grammar] A \emph{context-free grammar} is a quadruple $G=(\nonterminals,\terminals,\rules,S)$ where $\nonterminals$ is the set of non-terminal symbols, $\terminals$ is the set of terminal symbols (alphabet), $\rules$ is the set of production rules in the form $\Rule{A}{\alpha}$, for $A \in \nonterminals$ and $\alpha \in (\nonterminals \cup \terminals)^*$, and $S \in \nonterminals$ is the start symbol.
\end{definition}

%Given a non-terminal symbol $A$, we denote by $\L(A)$ the language of $A$, which is the set of terminal strings produced by all derivations starting with $A$:
%$$\L(A) = \set{s ~|~ A \derives s}$$
%The language of a grammar corresponds to the language of its start symbol.

%We denote by $\L(A)^{-1}$ the inverse of $\L(A)$, such that:
%$$a_1 a_2 \dots a_k \in \L(A) \iff a_k^{-1} a_{k-1}^{-1} \dots a_1^{-1} \in \L(A)^{-1}$$
%We define the inverse of $A$, denoted by $A^{-1}$, as a non-terminal symbol with the inverted rules of $A$.
%An inverted rule is defined as:
%$$(\Rule{A}{\alpha_1~\alpha_2~\dots~\alpha_k})^{-1} = \Rule{A^{-1}}{\alpha_k^{-1}~\alpha_{k-1}^{-1}~\dots~\alpha_1^{-1}}$$
%where $\alpha_i \in \terminals \cup \nonterminals$.
%If $\alpha_i$ is a non-terminal, we need to recursively invert its rules in the same way.
%The new non-terminal symbol $A^{-1}$, then, generates the language $\L(A^{-1}) = \L(A)^{-1}$.

%\subsection{Context-Free Path Queries}

We are interested in querying graph databases, represented using RDF.
An RDF graph is made of resources and the relationships between them.
A resource may be in one of the following pairwise disjoint sets:
\begin{itemize}
\item Internationalized Resource Identifiers (IRIs), which are an extension of Uniform Resource Identifiers (URIs) with support to a wider range of Unicode characters.
IRIs uniquely identify resources such as documents, movies or users' profiles in social networks;
\item literals, which specify a literal value such as a text, number or date; or
\item blank nodes, which are equivalent to labeled null values.
\end{itemize}

The relationships between resources are expressed in the form of \emph{triples}.
A triple is denoted by $(s,p,o)$, where $s$ is the \emph{subject}, $p$ is the \emph{predicate} and $o$ is the \emph{object}.
The subject of a triple is either an IRI or a blank node; the predicate (also known as the \emph{property}) is an IRI; and the object is either IRI, a literal or a blank node.
A finite set of triples forms an RDF database, which corresponds to a graph.
%6.3 Graph Equivalence
%Two RDF graphs G and G' are equivalent if there is a bijection M between the sets of nodes of the two graphs, such that:

%    M maps blank nodes to blank nodes.
%    M(lit)=lit for all RDF literals lit which are nodes of G.
%    M(uri)=uri for all RDF IRI references uri which are nodes of G.
%    The triple ( s, p, o ) is in G if and only if the triple ( M(s), p, M(o) ) is in G'
% With this definition, M shows how each blank node in G can be replaced with a new blank node to give G'.
\begin{definition}[Graph]
A \emph{graph} is a set of triples in $V \times E \times V$, where $V$ is a set of vertices and $E$ is a set of edge labels.
In RDF, it is possible that $V \cap E \neq \emptyset$.
\end{definition}

We can specify paths inside a graph by adequately choosing a sequence of triples.
\begin{definition}[Path and Trace]
A \emph{path} is a sequence of triples $(t_1,t_2,..t_k)$ from a given graph, where $t_i = (s_i,p_i,o_i)$, such that $o_i = s_{i+1}$.
The \emph{trace} of a path is the string formed by the concatenation of the edge labels $p$ from its triples.
The set of paths between two vertices $x$ and $y$ is denoted by $\paths(x,y)$. 
Notice that this includes the empty path between one node and itself.
Given a set of paths $\Pi \subseteq \paths(x, y)$, the set of traces defined by these paths is denoted as $\traces(\Pi)$.

%
%Given two vertices $x$ and $y$, we denote the possibly infinite set of all paths starting at $x$ and ending at $y$ by
%\label{def:helper-functions}
%$$paths(x,y) = \set{\pi ~|~ \pi = ((x,e_1,\_),...,(\_,e_k,y)) \in D}$$
%Additionally, $((x,\epsilon,x)) \in \paths(x,x)$.
%
%Given a set of paths $\Pi$, we denote the set of traces corresponding to those paths by
%$$\traces(\Pi) = \set{e_1 \dots e_k ~|~ ((x,e_1,\_),...,(\_,e_k,y)) \in \Pi}$$
\end{definition}

\begin{definition}[Context-Free Path Query]
\label{def:cfpq}
Given a data graph $D$ and a context-free grammar $G$, a \emph{context-free path query} $Q$ is a set of query pairs $(x,A)$ where $x$ is a vertex of the graph and $A$ a non-terminal symbol from a given grammar.
The evaluation of a context-free path query $Q$ produces the set of all vertexes $y$ such that there exists a path from $x$ to $y$ whose trace $s$ is derivable by $A$.
$$Eval(Q) = \set{y ~|~ \exists s\ .\ A \derives s \wedge s \in \traces(\paths(x,y))}$$
\end{definition}

%\subsection{Comments on the Notation Used in This Paper}
%\ciro{Não é necessário, mas talvez ajude no entendimento.}
%
%\ciro{Achei essa tabela em outro artigo. Acho que ficaria melhor do que uma sub-seção.}
%\begin{figure}
%    \centering
%    \includegraphics[width=.5\textwidth]{notation}
%    \caption{Example Table of Notation}
%    \label{fig:my_label}
%\end{figure}

\section{Context-Free Path Query Evaluation}

The next definition establishes the set of vertices that are reachable from a given vertex, by following a path represented by a string of (terminal and non-terminal) symbols of a grammar.

\begin{definition}[$G$-Reachable vertices]
\label{def:relation}
Let $G=\tuple{\nonterminals,\terminals,P, S}$ be a grammar, and $D \subseteq V\times E \times V$ be a data graph.
Given a vertex $x\in V$ and a string $\alpha \subseteq (\terminals \cup \nonterminals)^*$, the function $\R{x}{\alpha}$ defines the set of vertices reachable from $x$ by following an $\alpha$-derivable path in $D$:
$$\R{x}{\alpha}\ :\ V \times (\terminals \cup \nonterminals \cup \set{\epsilon})^* \mapsto \pow{V}.$$

\noindent This function is recursively defined on $\alpha$, as follows:
\begin{enumerate}
\item For $\alpha = \varepsilon$ (the empty string), each vertex is reachable from itself:  $\R{x}{\varepsilon} = \set{x}$. \label{def:relation-empty}

\item For $\alpha = p \in \terminals$, the set of vertices reachable from $x$ via a $p$-labeled edge is $\R{x}{p} = \set{y ~|~ (x,p,y) \in D}$. \label{def:relation-terminal}

\item
If $\alpha = A \in \nonterminals$, the set of vertices reachable from $x$ is defined by using the right-hand side of the productions of $A$ in $G$:  
$$\R{x}{A} = \bigcup_{\Rule{A}{\alpha} \in \rules} \R{x}{\alpha}.$$
\label{def:relation-nonterminal}
    
\item
If $\alpha=\alpha_1 \alpha_2$, the set of vertices reachable from $x$ is defined as:
$$\R{x}{\alpha_1 \alpha_2} = \bigcup_{w \in \R{x}{\alpha_1}} \R{w}{\alpha_2}.$$ 
\label{def:relation-string}
\end{enumerate}
It is easy to verify that this function is associative, since string concatenation and set union are both  associative operations.
\end{definition}

The following property establishes that for any vertex $y$, $G$-reachable from $x$, there exists a path in the graph whose labels form a string generated by the grammar $G$.

\begin{proposition}[Derivation of traces for paths in the graph]
\label{prop:relations}

Given a grammar $G=\tuple{\nonterminals,\terminals,P, S}$, a data graph $D \subseteq V\times E$, two nodes $x, y\in V$ and a string $\alpha \subseteq (\terminals \cup \nonterminals)^*$, we have that $y$ is in $\R{x}{\alpha}$ if and only if there is a $\alpha$-derivable path in $D$ from $x$ to $y$:
\begin{align*}
\forall x,y,\alpha.\  (y \in \R{x}{\alpha} \iff \exists s\ .
        &  \alpha \derives s\  \wedge \\ %\label{prop:relations} \\
        &s \in \traces(\paths(x,y))) \notag
\end{align*}

\begin{proof}% [Proposition~\ref{prop:relations}]
Assuming $s=p_1...p_m$, we proceed by induction on two variables, $m$ and $n$, representing respectively the length of the string $s$ and the number of steps in the derivation $\alpha \derivesk{n} p_1 ... p_m$.
\begin{itemize}
\item Base case (with $n=0, m=0$): \label{it:bc}

In this case, $\alpha \derivesk{0} \epsilon$ and $s = \alpha = \epsilon$.
By \defref{relation}{empty}, we also know that $y=x$ since $y \in \R{x}{\epsilon} = \set{x}$.
We need to show that
\begin{align*}
\forall x.\  (x \in \R{x}{\epsilon} \iff \epsilon \in \traces(\paths(x,x))) 
\end{align*}

This is straightforward since $\epsilon \in \traces(\paths(x,x))$.

Notice that when $m=0$, we have to build derivations from the empty string. 
So, $m=0 \implies n=0$.

\item Inductive step on $m$ (with $n=0$):
In this case, we have that $s = \alpha$, so we must prove that
\begin{align*}
\forall x,y,s.\  (y \in \R{x}{s} \iff   s \in \traces(\paths(x,y))) 
\end{align*}

This follows by mathematical induction on $m$.

% \begin{flalign}
% \forall x,y,p_{1 \dots m} (y \in \Res(x,p_{1 \dots m}) \iff 
% p_{1 \dots m} \in \traces(\paths(x,y)))  \label{eq:ind-step-m}
% \end{flalign}
% for an arbitrary $m$ (we still omit $p_{1 \dots m} \derivesk{0} p_{1 \dots m}$ since it is trivial).

% We assume as induction hypothesis that
% \begin{flalign}
% &\forall x,w,p_{1 \dots m-1} (w \in \Res(x,p_{1 \dots m-1}) \iff \notag \\
% &\exists p_{1 \dots m-1}\ .\ p_{1 \dots m-1} \in \traces(\paths(x,w))) \label{eq:ih-m}
% \end{flalign}

% When demonstrating the ``if'' direction, we know that $y \in \Res(x,p_{1 \dots m})$, which, by \defref{relation}{string}, is $\Res(x,p_{1\dots m}) = \bigcup_{w \in \Res(x,p_{1\dots m-1})} \Res(w,p_m)$ and $\Res(w,p_m) = \set{y | (w,p_m,y) \in D}$.
% By the induction hypothesis~(\ref{eq:ih-m}), we know there exists a path $\pi \in \paths(x,w)$ such that $\traces(\set{\pi}) = \set{p_{1\dots m-1}}$.
% The concatenation of $\pi$ with a triple $(w,p_m,y)$ produces a new path $\pi' \in \paths(x,y)$ such that $\traces(\set{\pi'}) = \set{p_{1\dots m}}$, so $p_{1\dots m} \in \traces(\paths(x,y))$.

% On the other hand, when demonstrating the ``only if'' direction, we know that $\pi'$ defined above is in $\traces(\paths(x,y))$.
% By \defref{relation}{string}, $\Res(x,p_{1\dots m}) = \bigcup_{w \in \Res(x,p_{1\dots m-1})} \Res(w,p_m)$.
% The induction hypothesis~(\ref{eq:ih-m}) warranties that $w \in \Res(x,p_{1 \dots m-1})$.

% Therefore $y \in \Res(x,p_{1\dots m})$, and \ref{eq:ind-step-m} holds.

\item Inductive step on $n$ (with $m > 0$):
We need to demonstrate that 

\begin{align*}
\forall x,y,\alpha.\  (y \in \R{x}{\alpha} \\ \iff &\exists\  p_1...p_m\ .
          \alpha \derivesk{n} p_1...p_m \\
        &\wedge p_1...p_m \in \traces(\paths(x,y))) 
\end{align*}
for an arbitrary $n$.

Since $n>0$, we have that $\alpha = \alpha_1\ A\ \alpha_2$, where $A\in \nonterminals$ and $\alpha_1, \alpha_2 \in (\nonterminals \cup \terminals)^*$. 
By Induction Hypothesis, we have that there exist vertices $v, w \in V$ and indexes $k, j$ where $0 \leq k \leq j \leq m$ such that:

\begin{align*}
v \in \R{x}{\alpha_1} \iff &
          \alpha_1 \derives p_1...p_k \\
        &\wedge p_1...p_k \in \traces(\paths(x,v)) \\[2mm]
w \in \R{v}{A} \iff &
          A \derives p_{k+1}...p_j \\
        &\wedge p_{k+1}...p_j \in \traces(\paths(v,w)) \\[2mm]
y \in \R{w}{\alpha_2} \iff & \alpha_2 \derives p_{j+1}...p_m \\
        &\wedge p_{j+1}...p_m \in \traces(\paths(w,y)) 
\end{align*}
These hypotheses, together with \defref{relation}{string} allow us to conclude the proof.

\end{itemize}
%Thus, Proposition~\ref{prop:relations} holds for all $m$ and $n$.
\end{proof}

\end{proposition}

%The proof of this proposition is by induction on the size of $s$ and the number of steps of the derivation $\alpha \derives s$. 
%The proof is presented in Appendix~\ref{sec:proofs} \ciro{Appendices are not allowed!}.

% \bigskip 
% The next section presents an algorithm to implement the relations defined in Definition~\ref{def:relation}.

\subsection{Our Algorithm}

In this section we present our proposal for the evaluation of CFPQs.
Our algorithm receives a grammar, a data graph and a query, and follows context-free paths inside the data graph.
The goal of the algorithm is to identify pairs of vertices linked by paths whose traces are strings generated by the grammar.

The following example illustrates the problem:
\begin{example}
\label{ex:example1}
Let us consider a grammar $G$ with the following production rules: % nomeei a gramática aqui pra referenciar no exemplo 3.5
\[\Rule{S}{a\ S\ b} \qquad \Rule{S}{\varepsilon}
\]
and the data graph given in Figure~\ref{fig:exemplo1}.
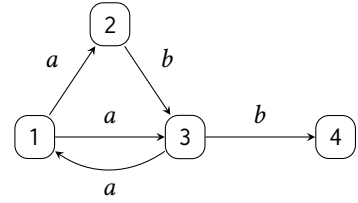
\begin{figure}%[b]
\centering
\begin{tikzpicture}[level/.style={sibling distance=20mm/#1}, inner sep=5pt]

\node[draw, rectangle, rounded corners=5pt] at (0,0) (1) {\texttt{1}};
\node[draw, rectangle, rounded corners=5pt] at (1,1.5) (2) {\texttt{2}};
\node[draw, rectangle, rounded corners=5pt] at (2,0) (3) {\texttt{3}};
\node[draw, rectangle, rounded corners=5pt] at (4,0) (4) {\texttt{4}};

\draw [->, >=stealth] (1) to node [auto] {$a$} (2);
\draw [->, >=stealth] (1) to node [auto] {$a$} (3);
\draw [->, >=stealth] (2) to node [auto] {$b$} (3);
\draw [->, >=stealth, bend left=35] (3) to node [auto] {$a$} (1);
\draw [->, >=stealth] (3) to node [auto] {$b$} (4);

\end{tikzpicture}
\caption{Example Graph.}
\label{fig:exemplo1}
\end{figure}

Given the query $Q=\set{\tuple{1, S}, \tuple{3, S}}$, our algorithm goes through paths starting at vertices 1 and 3 whose trace is generated by $S$. %  looks for vertices that can be reached from vertices 1 and 3, via paths whose traces are generated by $S$.
In this way all the production rules of $S$ will be investigated for paths starting at each of these vertices.

For the query $Q$, our algorithm will compute the sets of vertices $\set{1, 3, 4}$, reachable from node 1, and the set $\set{3, 4}$, reachable from node 3.
~\hfill$\diamond$
\end{example}

Our method relies on two assumptions: 
\textit{(i)} there may be several paths starting at a given node of the data graph; and
\textit{(ii)} for each of these paths, their trace may be derivable from a non-terminal of the grammar.

Our algorithm explores these two properties to parse all the paths from a given vertex, in order to discover which of them have traces derivable by a given non-terminal.
The parsing of all these traces is performed in an incremental way.
In our setting, a query $Q$ is represented by a set of pairs $\tuple{v, A}$, where $v$ is a vertex of the data graph and $A$ is a non-terminal symbol of the grammar.
For each pair $\tuple{v, A}$ of the query, our  algorithm identifies all the paths from $v$ whose traces are strings derivable from $A$.

\bigskip

In a traditional parsing setting, we may use the notion of \textit{grammar item} to guide the parsing process.
Grammar items use a dot on the right-hand side of a production rule to mark the progress of the parsing.
Traditional parsing techniques are tailored to process one input string at a time.
The information carried by the dot is related just to the progress of the parsing.
In our case, we also need to identify the strings that form paths of the graph being parsed. 
Thus, we associate vertices of the graph to the positions of the parsing process.
In our case, we will use sets of vertices of the graph within the items, in the place where the dot may appear.
The next definition captures this idea:

\begin{definition}[Trace Item] \label{def:trace-item} 
Given a context-free grammar $G = \tuple{N, \Sigma, P, S}$ and a data graph $D \subseteq V \times E \times V$, a \textit{Trace Item} is a pair formed by a production rule and a function associating a set of graph nodes to each position of the right-hand side of the rule.
Formally, a trace item is defined as the pair
$\tuple{\Rule{A}{\alpha}, f}$,
where $\Rule{A}{\alpha}\in P$ and $f: \{0,\dots,|\alpha|\} \rightarrow \pow{V}$.

The trace item  $\tuple{\Rule{A}{\alpha_1,\dots,\alpha_n}, f}$, where $f=\{0\mapsto C_0,\dots,n\mapsto C_n\}$ will be noted as
$[\Rule{A}{C_0 ~\alpha_1~ C_1 ~... ~\alpha_n~ C_n}]$.
The sets $C_1,\dots,C_n$ will be called \emph{position sets}.
~\hfill$\diamond$
\end{definition}

In general, given position sets $C_1, C_2$ and a grammar symbol $\alpha$, a sequence $C_1\ \alpha\ C_2$ in the right-hand side of an item indicates that each vertex in $C_2$ will be reached by an $\alpha$-derivable path beginning at a vertex in $C_1$.
For instance, the trace item $[\ \Rule{S}{\set{1}\ a\ \set{2, 3}\ S\ \set{~}\ b\ \set{~}}\ ]$ in Example~\ref{ex:example1}, indicates that the parsing process is in a stage where $a$-derivable paths linking vertex 1 to vertices 2 and 3 in the data graph have been identified.

Next, we present the intuitive idea of our algorithm.
In order to solve a query $Q$, our algorithm will start processing trace items obtained from the query pairs and rules of the grammar: for each query pair $(v,A) \in Q$, we create one trace item for each production rule of $A$ with $v$ in its first position set.
We will use special marks $\u$ and $\m$ for unprocessed and processed vertices inside position sets, respectively, in order to keep track of what vertices have already been processed\footnote{We omit the $\m$ and $\u$ marks from vertices in position sets when such distinction is unnecessary.}.
Our algorithm will process trace items until there are no unprocessed vertices belonging to any position set.

The next example shows how to compute the answers for the given query, graph and grammar.

%Such marks are important because the order of processing the paths is arbitrary.

%For each pair $(v, A)$ in the query we will look for paths starting at $v$, for each rule of $A$.

\begin{example}
\label{ex:example2}%[Example~\ref{ex:example1} continued]
Given the query $Q=\set{\tuple{1, S}, \tuple{3, S}}$ and data graph $D$ and grammar $G$ from Example~\ref{ex:example1}, we start the parsing process by creating trace items.
For each query pair $(v,A) \in Q$, we create one trace item for each production rule of $A$ with $v$ in its first position set.
%Such marks are important because the order of processing the paths is arbitrary.
For the query $Q$ we build the trace items:
\begin{eqnarray}
    {[}\ S&\rightarrow& \set{1\u}\ a\ \set{~}\ S\ \set{~}\ b\ \set{~}\ {]}\label{eq:ex2:1}\\
    {[}\ S&\rightarrow& \set{1\u}\ {]}\label{eq:ex2:2}\\
    {[}\ S&\rightarrow& \set{3\u}\ a\ \set{~}\ S\ \set{~}\ b\ \set{~}\ {]}\label{eq:ex2:3}\\
    {[}\ S&\rightarrow& \set{3\u}\ {]} \label{eq:ex2:4}
\end{eqnarray}

Our algorithm picks the unprocessed vertices in an arbitrary order.
Let us start with vertex 1 from trace item~(\ref{eq:ex2:1}).
This vertex appears in a position set before the terminal symbol $a$.
We must walk from vertex 1 to all its neighbors linked by an $a$-labeled edge in $D$.
The neighbors vertices 2 and 3 must then be added to the next position set in the trace item.
Doing so, our item will become $[\ \Rule{S}{\set{1\m}\ a\ \set{2\u,3\u}\ S\ \set{~}\ b\ \set{~} }\ ]$.
Notice that vertex 1$\u$ has changed to 1$\m$ to signal that this vertex has been processed. 
New vertices are added as unprocessed by using the mark $\u$.
%This means that vertex 1 have already been processed and now we have two new vertices, 2 and 3, to process.
Now we may pick vertex 2 for the next step.
This vertex is in a position set before the non-terminal symbol $S$.
That indicates that we have to look for $S$-derivable paths starting at vertex 2.
We build the following new items:
\begin{eqnarray}
    {[}\ S&\rightarrow& \set{2\u}\ a\ \set{~}\ S\ \set{~}\ b\ \set{~}\ {]} \label{eq:ex2:5}\\
    {[}\ S&\rightarrow& \set{2\u}\ {]} \label{eq:ex2:6}
\end{eqnarray}

Now item~(\ref{eq:ex2:1}) becomes $[\ \Rule{S}{\set{1\m}\ a\ \set{2\m,3\u}\ S\ \set{~}\ b\ \set{~} }\ ]$ and we have to pick another vertex to process.
Picking vertex 2 from item~(\ref{eq:ex2:5}) we verify that there is no $a$-labeled edge going from vertex 2 to any other vertex in the graph.
That means that there is no $a$-derivable path from this vertex.
Item~(\ref{eq:ex2:5}) then becomes $[\ \Rule{S}{\set{2\m}\ a\ \set{}\ S\ \set{}\ b\ \set{}}\ ]$.

Let us now pick vertex 2 from item~(\ref{eq:ex2:6}).
This item was built from an $\epsilon$-rule. 
As the vertex 2 belongs to the first and last position set of this item, that means that there is a $S$-derivable path from vertex 2 to itself (the empty path).
So, we augment the data graph with an $S$-labelled edge (shown in \textbf{boldface}): 
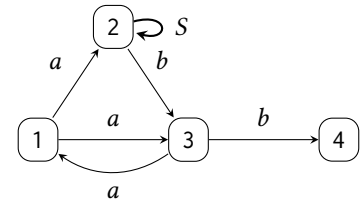
\begin{center}
\begin{tikzpicture}[level/.style={sibling distance=20mm/#1}, inner sep=5pt]

\node[draw, rectangle, rounded corners=5pt] at (0,0) (1) {\texttt{1}};
\node[draw, rectangle, rounded corners=5pt] at (1,1.5) (2) {\texttt{2}};
\node[draw, rectangle, rounded corners=5pt] at (2,0) (3) {\texttt{3}};
\node[draw, rectangle, rounded corners=5pt] at (4,0) (4) {\texttt{4}};

\draw [->, >=stealth] (1) to node [auto] {$a$} (2);
\draw [->, >=stealth] (1) to node [above] {$a$} (3);
\draw [->, >=stealth] (2) to node [auto] {\!\!$b$} (3);
\draw [->, >=stealth] (3) to node [auto] {$b$} (4);
\draw [->, >=stealth, bend left=35] (3) to node [below] {$a$} (1);

% loops
\path (2) edge [loop right,->, >=stealth, thick] node {$S$}  (2);
\end{tikzpicture}
\end{center}

Now, item~(\ref{eq:ex2:6}) becomes $[\ \Rule{S}{\set{2\m}}\ ]$.
The addition of the new, $S$-labelled edge to the data graph triggers a modification to the existing items: 
we add the unprocessed vertex 2 to any position set $C$ appearing in a trace item matching the pattern $[\ \dots \set{2,\dots}\ S\ C\ \dots \ ]$.
In our case, item~(\ref{eq:ex2:1}) becomes $[\ \Rule{S}{\set{1\m}\ a\ \set{2\m,3\u}\ S\ \set{2\u}\ b\ \set{~} }\ ]$.

We may now pick the newly added vertex 2$\u$ in item~(\ref{eq:ex2:1}).
Now we have a vertex in a position set before the terminal $b$.
As we did before, we look for $b$-labeled edges going out from 2 in the data graph.
There is only one such edge, which arrives at vertex 3.
Item~(\ref{eq:ex2:1}) then becomes $[\ \Rule{S}{\set{1\m}\ a\ \set{2\m,3\u}\ S\ \set{2\m}\ b\ \set{3\u} }\ ]$.

Now we pick the newly added vertex 3 in the last position set of item~(\ref{eq:ex2:1}).
As this vertex is at the last position set of the item, we infer that there is an $S$-valid path from vertex 1 to vertex 3.
As $\tuple{1, S} \in Q$, we have found one answer for our query.
Item~\ref{eq:ex2:1} then becomes $[\ \Rule{S}{\set{1\m}\ a\ \set{2\m,3\u}\ S\ \set{2\m}\ b\ \set{3\m} }\ ]$.
Then, the data graph is augmented with a new $S$-labelled edge from 1 to 3:
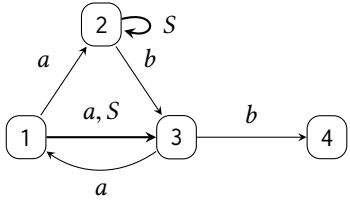
\begin{center}
\begin{tikzpicture}[level/.style={sibling distance=20mm/#1}, inner sep=5pt]

\node[draw, rectangle, rounded corners=5pt] at (0,0) (1) {\texttt{1}};
\node[draw, rectangle, rounded corners=5pt] at (1,1.5) (2) {\texttt{2}};
\node[draw, rectangle, rounded corners=5pt] at (2,0) (3) {\texttt{3}};
\node[draw, rectangle, rounded corners=5pt] at (4,0) (4) {\texttt{4}};

\draw [->, >=stealth] (1) to node [auto] {$a$} (2);
\draw [->, >=stealth,thick] (1) to node [above] {$a, S$} (3);
\draw [->, >=stealth] (2) to node [auto] {\!\!$b$} (3);
\draw [->, >=stealth] (3) to node [auto] {$b$} (4);
\draw [->, >=stealth, bend left=35] (3) to node [below] {$a$} (1);

% loops
\path (2) edge [loop right,->, >=stealth, thick] node {$S$}  (2);
\end{tikzpicture}
\end{center}

This process is repeated until there are no more unprocessed vertices.
The complete step-to-step process is presented in Table~\ref{tab:example2-complete}.
That will result in the following set of items:
\[\begin{array}{l}
\Item{S}{\set{1\m} ~a~ \set{2\m,3\m} ~S~ \set{2\m,3\m,4\m} ~b~ \set{3\m,4\m}}, \ \ 
\Item{S}{\set{1\m}}, \\
\Item{S}{\set{2\m}~a~ \emptyset ~S~ \emptyset ~b~  \emptyset}, \qquad\qquad\qquad\ \ 
\Item{S}{\set{2\m}}, \\
\Item{S}{\set{3\m} ~a~ \set{1\m} ~S~ \set{1\m,3\m,4\m} ~b~  \set{4\m}}, \qquad 
\Item{S}{\set{3\m}}
% Obs.: não são criados items começando em 4 porque não tá na query Q e ninguém começa nenhuma derivação com S lá
\end{array}
\]

\begin{table}[htpb]
\newcounter{actioncounter}
\newcommand{\nxtact}{%
	\stepcounter{actioncounter}%
	\theactioncounter} 

\centering \small
\begin{tabular}{@{}cll@{}}
\toprule
\textbf{\#}            & \textbf{Operation}                          & \textbf{Updated items}        \\ \midrule

\nxtact  & line~\ref{l:let-I} & \Item{S}{\set{\underline{1\u}}\ a\ \set{~}\ S\ \set{~}\ b\ \set{~}}, \Item{S}{\set{\underline{1\u}}}, \\
& & \Item{S}{\set{\underline{3\u}}\ a\ \set{~}\ S\ \set{~}\ b\ \set{~}}, \Item{S}{\set{\underline{3\u}}}
\\ \midrule

\nxtact & line~\ref{l:advance} & \Item{S}{\set{\underline{1\m}}\ a\ \set{\underline{2\u},\underline{3\u}}\ S\ \set{~}\ b\ \set{~}}
\\ \midrule

\nxtact & line~\ref{l:update-I} & \Item{S}{\set{1\m}\ a\ \set{\underline{2\m},3\u}\ S\ \set{~}\ b\ \set{~}}, \\
& & \Item{S}{\set{\underline{2\u}}\ a\ \set{~}\ S\ \set{~}\ b\ \set{~}}, \Item{S}{\set{\underline{2\u}}}
\\ \midrule

\nxtact & line~\ref{l:advance} & \Item{S}{\set{\underline{2\m}}\ a\ \set{~}\ S\ \set{~}\ b\ \set{~}}
\\ \midrule

\nxtact & lines~\ref{l:update-D'}, \ref{l:notify} & \Item{S}{\set{\underline{2\m}}}, \\
& & \Item{S}{\set{1\m}\ a\ \set{2\m,3\u}\ S\ \set{\underline{2\u}}\ b\ \set{~}}
\\ \midrule

\nxtact & line~\ref{l:advance} & \Item{S}{\set{1\m}\ a\ \set{2\m,3\u}\ S\ \set{\underline{2\m}}\ b\ \set{\underline{3\u}}}
\\ \midrule

\nxtact & lines~\ref{l:update-D'}, \ref{l:notify} & \Item{S}{\set{1\m}\ a\ \set{2\m,3\u}\ S\ \set{2\m}\ b\ \set{\underline{3\m}}}
\\ \midrule

\nxtact & lines~\ref{l:update-D'}, \ref{l:notify} & \Item{S}{\set{\underline{1\m}}}
\\ \midrule

%\nxtact & & \Item{S}{\set{1\m}\ a\ \set{2\m,3\u}\ S\ \set{2\m}\ b\ \set{\underline{3\m}}}
%\\ \midrule

\nxtact & lines~\ref{l:update-D'}, \ref{l:notify} & \Item{S}{\set{\underline{3\m}}} \\
& & \Item{S}{\set{1\m}\ a\ \set{2\m,3\u}\ S\ \set{2\m,\underline{3\u}}\ b\ \set{3\m}}
\\ \midrule

\nxtact & line~\ref{l:advance} & \Item{S}{\set{1\m}\ a\ \set{2\m,3\u}\ S\ \set{2\m,\underline{3\m}}\ b\ \set{3\m,\underline{4\u}}}
\\ \midrule

\nxtact & lines~\ref{l:update-D'}, \ref{l:notify} & \Item{S}{\set{1\m}\ a\ \set{2\m,3\u}\ S\ \set{2\m,3\m}\ b\ \set{3\m,\underline{4\m}}}
\\ \midrule

\nxtact & line~\ref{l:advance} & \Item{S}{\set{\underline{3\m}}\ a\ \set{\underline{1\u}}\ S\ \set{~}\ b\ \set{~}}
\\ \midrule

\nxtact & line~\ref{l:advance} & \Item{S}{\set{3\m}\ a\ \set{\underline{1\m}}\ S\ \set{\underline{1\u},\underline{3\u},\underline{4\u}}\ b\ \set{~}}
\\ \midrule

\nxtact & line~\ref{l:advance} & \Item{S}{\set{3\m}\ a\ \set{1\m}\ S\ \set{1\u, 3\u,\underline{4\m}}\ b\ \set{~}}
\\ \midrule

\nxtact & line~\ref{l:advance} & \Item{S}{\set{3\m}\ a\ \set{1\m}\ S\ \set{1\u, \underline{3\m},4\m}\ b\ \set{\underline{4\u}}}
\\ \midrule

\nxtact & lines~\ref{l:update-D'}, \ref{l:notify} & \Item{S}{\set{3\m}\ a\ \set{1\m}\ S\ \set{1\u, 3\m,4\m}\ b\ \set{\underline{4\m}}}
\\ \midrule

\nxtact & line~\ref{l:advance} & \Item{S}{\set{3\m}\ a\ \set{1\m}\ S\ \set{\underline{1\m}, 3\m,4\m}\ b\ \set{4\m}}
\\ \midrule

\nxtact & line~\ref{l:update-I} & \Item{S}{\set{1\m}\ a\ \set{2\m,\underline{3\m}}\ S\ \set{2\m,3\m,\underline{4\u}}\ b\ \set{3\m,4\m}}
\\ \midrule

\nxtact & line~\ref{l:advance} & \Item{S}{\set{1\m}\ a\ \set{2\m,3\m}\ S\ \set{2\m,3\m,\underline{4\m}}\ b\ \set{3\m,4\m}}
%\\ \midrule

\\ \bottomrule
\end{tabular}
\normalsize

\caption{Step-by-step behavior of Algorithm~\ref{alg:shared-prefixes}.}
\label{tab:example2-complete}
\end{table}

%Notice that the order in which unprocessed vertices are chosen is not relevant to find a solution for the query.

The solutions computed by our algorithm are shown as \textbf{bold} arrows, labeled by non-terminals, in Figure~\ref{fig:example-result}.
~\hfill$\diamond$

\begin{figure}[h]
\centering
\begin{tikzpicture}[level/.style={sibling distance=20mm/#1}, inner sep=5pt]

\node[draw, rectangle, rounded corners=5pt] at (0,0) (1) {\texttt{1}};
\node[draw, rectangle, rounded corners=5pt] at (1,1.5) (2) {\texttt{2}};
\node[draw, rectangle, rounded corners=5pt] at (2,0) (3) {\texttt{3}};
\node[draw, rectangle, rounded corners=5pt] at (4,0) (4) {\texttt{4}};

\draw [->, >=stealth] (1) to node [auto] {$a$} (2);
\draw [->, >=stealth, thick] (1) to node [above] {$a$,\ $S$} (3);
\draw [->, >=stealth] (2) to node [auto] {\!\!$b$} (3);
\draw [->, >=stealth, thick] (3) to node [auto] {$b$,\ $S$} (4);

\draw [->, >=stealth, bend left=35] (3) to node [below] {$a$} (1);
\draw [->, >=stealth, bend right=60, thick] (1) to node [auto] {$S$} (4);

% loops
\path (1) edge [loop left,->, >=stealth, thick] node {$S$}  (1);
\path (2) edge [loop right,->, >=stealth, thick] node {$S$}  (2);
\path (3) edge [loop above,->, >=stealth, thick] node {$S$}  (3);
%\path (4) edge [loop right,->, >=stealth, thick] node {$S$}  (4);

\end{tikzpicture}
\caption{Result graph for the query of Example~\ref{ex:example1}.}
\label{fig:example-result}
\end{figure}
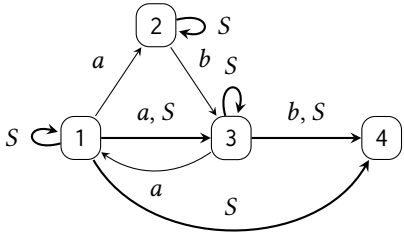
\end{example}

Let us now present our algorithm for processing context-free path queries (Algorithm~\ref{alg:shared-prefixes}). 
Our technique is based on the idea of building and updating a set of trace items.
The input parameters of the algorithm are:
\begin{enumerate}
\item A context-free grammar $G=\tuple{\nonterminals,\Sigma,P,S}$, defined by the user.
%No restriction (such as normal forms) is needed for the grammar.

\item An RDF graph $D=V\times \Sigma \times V$ with edges restricted to the grammar alphabet.% and enhanced with explicit inverted edges.

\item A set of query pairs $Q \subseteq V \times \nonterminals$.
Each pair of the query set indicates a start vertex and non-terminal symbol used for recognizing paths.
% It is often initialized as $Q = V \times \set{S}$, where $S$ is the start symbol of the given grammar.
% This indicates that the algorithm must walk through paths constrained to the language of the grammar's start symbol   starting at every vertex in the graph.
\end{enumerate}

\begin{algorithm}[htb]
    \KwIn{$G = (\nonterminals,\terminals,\rules, S), ~Q \subseteq V \times \nonterminals, ~D \subseteq V\times \terminals \times V $}

\KwOut{$D' \subseteq V\times \terminals \times V$}

% \SetKwProg{Procedure}{procedure}{}{end}
\SetKwProg{Function}{function}{}{end}
\Function{eval}{
    $\I := \set{[\Rule{A}{\set{w\u} ~\alpha_1~ \emptyset ~... ~\alpha_n~ \emptyset }] ~|~ \Rule{A}{\alpha_1~ ... ~\alpha_n} \in \rules \wedge (w,A) \in Q}$
    \label{l:start-let} \label{l:let-I} \\
    $D' := D$ \label{l:let-D'} \label{l:end-let}
    
    \While{$\exists\ i,x \st i=[\Rule{A}{...~ \set{x\u, ...}~...}] \in \I $
    }{ 
        \label{l:main-loop}
        \inlineSwitch{i}{ \\
        \uCase{$i = [\Rule{A}{... ~\set{x\u, ...} ~\alpha_k ~C_k ~...}]$}{
            \label{l:case-mid}
            \uIf{$\alpha_k \in \terminals \vee [\Rule{\alpha_k}{\set{x} \dots}] \in \I$}{\label{l:in-Res}
                $C_k := C_k ~\munion~ \set{y{\u} ~|~ (x,\alpha_k,y) \in D'} $ \label{l:advance}
            } \Else { 
                  \label{l:notin-Res}
                $\I := \I ~\cup~ \set{[\Rule{\alpha_k}{\set{x\u} ~\beta_1 ~\set{~} ~... ~\beta_n ~\set{~}}] ~|~ \Rule{\alpha_k}{\beta_1 ~...~ \beta_n} \in \rules}$ \label{l:update-I}
            }
            \label{l:end-case-mid}
        }\Case{$i = [\Rule{A}{\set{w} \dots \set{x\u, ...}}]$}{
            \label{l:case-end}
            $D' := D' \cup \set{(w,A,x)}$  \label{l:update-D'}
            
            \ForEach{$[\Rule{B}{\dots \set{w\m,...} ~A~C \dots}] \in \I$}{ \label{l:foreach}
                $C := C ~\munion~ \set{x\u}$  \label{l:notify}
            }
        }
        \label{l:end-case-end}
        }
        $mark(x, i)$ \label{l:mark} 

    }
    \Return $D'$ \label{l:return}
}
	\caption{The Trace Item-based Algorithm}
	\label{alg:shared-prefixes}
\end{algorithm}

Our algorithm uses the $\munion$ operator to perform unions between sets of marked and unmarked vertices.
This operator is defined as follows: given the position sets $C$ and $\set{x\u}$, the union between them is defined as:
\begin{eqnarray*}
C \munion \set{x\u} &=& \left\{ \begin{array}{ll}
C,               & \mathit{if}\ x\m \in C\\
C\cup \set{x\u}, & \mathit{otherwise}
\end{array}\right.
\end{eqnarray*}
That is, if the vertex $x$ has already been processed, it is kept as processed in the position set. 
Otherwise, it is added as unprocessed.

\bigskip

The following data structures are manipulated during the algorithm's execution:
\begin{itemize}
\item[\I:] A set of trace items, iterativelly computed by the algorithm.
    
\item[$D'$:] A data graph $D'$, containing the original data graph $D$ incrementally augmented with  new, non-terminal-labeled edges.
\end{itemize}

Lines~\ref{l:start-let}-\ref{l:end-let} initialize \I\ and $D'$.
For each pair $(w, A) \in Q$ and rule $\Rule{A}{\alpha_1 ... \alpha_n \in \rules}$, the set \I\ is initialized with items $\Rule{A}{\set{w\u} \alpha_1 \emptyset ... \alpha_n \emptyset}$.
The graph $D'$ is initialized as a copy of the input graph $D$.
These steps prepare the algorithm to enter the main loop that processes unmarked vertices in items of \I.
The main loop concludes when there are no such unmarked vertices.

The processing of unmarked vertices is divided into two cases:
\begin{enumerate}
\item In the first case (lines~\ref{l:case-mid}-\ref{l:end-case-mid}), given the trace item $i = [\Rule{A}{C_0\ \alpha_1\ C_1\ ...\alpha_n\ C_n}]$, $x\u$ belongs to a position set $C_{k-1}$ that is not the last position set of the item.
%To advance the recognition, at line~\ref{l:advance} the algorithm adds to position set $C$ following $\alpha$ all vertices $y$ such that $(x,\alpha,y) \in D'$, which are the (possibly incomplete) results for the pair $(w,\alpha) \in H$.
%At this point it can occur that either:
\begin{enumerate}
\item If $\alpha_k \in \terminals$, we add to $C_k$ all the vertices $y\u$ such that there exists an edge $(x,\alpha_k,y) \in D'$  (line~\ref{l:let-D'}).

\item If $\alpha_k \in \nonterminals$ and $\Item{\alpha_k}{\set{w}\dots} \in \I$, we add all $y\u$ to $C_{k}$ such that there is an edge $(x,\alpha_k,y) \in D'$ (this case is also treated by line~\ref{l:let-D'}).

\item If $\alpha_k \in \nonterminals$ and there is no trace item $\Item{\alpha_k}{\set{x\ ...}\ ...}$, our algorithm initiates the search for $\alpha_k$-derivations beginning at $x$.
This is done by creating new trace items $\Rule{\alpha_k}{\set{x\u}\dots}$ and adding them to \I\ (line~\ref{l:update-I}).
\end{enumerate}

\item In the second case of the main loop, lines~\ref{l:case-end} to~\ref{l:end-case-end}, we identify that the vertex $x$ belongs to the last position set of a trace item.
The item $i = [\Rule{A}{\set{w} \dots \set{x\u, ...}}]$ states that we have walked a path from the vertex $w$ to $x$ in the data graph $D'$.
So, our algorithm generates a new $A$-labeled edge connecting these two vertices (line~\ref{l:update-D'}).
After this operation, we must update with $x\u$ all position sets $C$ such that $\Item{B}{\dots \set{w,\dots}\ A\ C \dots} \in \I$ (line~\ref{l:notify}).
\end{enumerate}

The vertex $x\u$ from the generalized item $i$ is marked as visited at the end of the loop body (line~\ref{l:mark}).
When there are no more unmarked vertices, the main loop stops and the decorated graph $D'$ is returned (line~\ref{l:return}).

\bigskip 

In the next sections we analyze the behaviour of our algorithm in terms of correctness and runtime and memory complexity.
% For example, let us suppose a grammar rule $\Rule{S}{A\ B}$.
% The generalized item $\Rule{S}{\set{a\m} ~A~ \set{b\u,c\u} ~B~ \set{~}}$ indicates that at some point of the computation, the vertex $a$ will be the departing point for an $S$-generated path. 
% The vertices in $\set{b, c}$ were reached by an $A$-generated path starting at $a$.
% The empty set at the end-point of the item indicates that (so far) there are no $B$-generated paths starting at $b$ or $c$.

% \begin{figure}
% \input{fig/transitions}
% \begin{eqnarray*}
% \Rule{S}{\set{v\u}}, \Res(v,S)=\set{\dots} &\leadsto& \Rule{S}{\set{v\m}, \Res(v,S) = \set{v,\dots}}
% \end{eqnarray*}
% \caption{Algorithm~\ref{alg:shared-prefixes}'s basic operations}
% \label{alg:shared-prefixes}
% \end{figure}

\subsection{Algorithm Correctness}
In this section, we show the correctness of our algorithm.

\begin{proposition}
\label{prop:I}
Let $G=\tuple{\nonterminals,\terminals,P, S}$ be a grammar, $D \subseteq V \times E \times V$ a data graph and a query pair $\tuple{w, A} \in Q$. 
Given $\Item{A}{\set{w} \alpha_1\ C_1 ... \alpha_j C_j ...} \in \I$ computed by Algorithm~\ref{alg:shared-prefixes}, then for any vertex $x \in V$ we have
$$x \in C_j \iff x \in \R{w}{\alpha_1 \dots \alpha_j}.$$
\end{proposition}

\begin{proof}[Sketch]
We analyze the behaviour of the algorithm at the lines that change the set $I$ of trace items:
\begin{description}
\item[(line~\ref{l:let-I})] The set \I\ is initialized to contain the item $\Item{A}{\set{w\u} ~\alpha_1~ \emptyset ~... ~\alpha_n~ \emptyset }$, for each rule $\Rule{A}{\alpha_1~ ... ~\alpha_n} \in \rules$.
From this construction we can see that for $j=0$, we have that $w=x$, $C_0 = \set{x} = \set{w}$ and $\alpha_1~ ... ~\alpha_j = \epsilon$. 
In this case, it is evident that %easy to see that
$$w \in C_0 \iff w \in \R{w}{\epsilon}.$$

\item[(line~\ref{l:update-I})] At this line, new trace items are added into the set \I\ for each rule~$\Rule{\alpha_k}{\beta_1 ... \beta_n}$.
The creation of new items is in under the same conditions presented at line~\ref{l:let-I}.
Again $j=0$, so we have $w=x$, $C_0 = \set{x} = \set{w}$ and $\beta_1~ ... ~\beta_j = \epsilon$. 
In this case, we have
$$w \in C_0 \iff w \in \R{w}{\epsilon}.$$

\item[(line~\ref{l:advance})] A position set $C$ in \I\ is incremented with new vertices $y$ such that $(x,\alpha_k,y) \in D'$.
We can distinguish two cases:

\begin{trivlist}
\item[-] If $\alpha_k$ is a terminal symbol, we add to $C_k$ all vertices $y$ such that exists a $\alpha_k$-labeled edge from $x$ to $y$ in $D'$:
$$y \in C_k \iff y \in \R{x}{\alpha_k}.$$
This condition holds by \defref{relation}{terminal}.

\item[-] If $\alpha_k \in \nonterminals$ we need to add to $C_k$ all the vertices $y$ such that there is an edge labelled $(x, \alpha_k, y)$ in $D'$.
Notice that this edge was the result of a previous processing, meaning that the algorithm has already discovered a path from $x$ to $y$ such that its trace corresponds to the right-hand side of a production rule of $\alpha_k$.
Thus,  
$$y \in C_k \iff y \in \R{x}{\alpha_k}.$$
This condition holds by \defref{relation}{nonterminal}.
\end{trivlist}

\item[(line~\ref{l:notify})] We deal with those vertices $x$ appearing at the last position set of a trace item $\Item{A}{\set{w\m} ... \set{x\u,...}}$ built from a production rule $\Rule{A}{\gamma}$. 
Items with this configuration indicate the existence of a path from $w$ to $x$ in $D'$ such that its trace is the string $\gamma$.
Our algorithm adds a new $A$-labeled edge from $w$ to $x$ (line~\ref{l:update-D'}), thus using the production rule.
Thus, for every item $i = \Item{B}{... \set{w\m,...}\ A\ C_j ...}$ built from a production rule $\Rule{B}{\gamma_1\ A\ \gamma_2}$, we can verify that:
$$x \in C_j \iff x \in \R{w}{A}.$$
This condition holds by Definitions~\ref{def:relation}.\ref{def:relation-nonterminal} and~\ref{def:relation}.\ref{def:relation-string}.
\end{description}
\end{proof}

We start by presenting evidences that the proposed algorithm is correct.

The result graph $D'$ is only updated at line~\ref{l:let-D'}, where it just copies the input graph $D$, and at line~\ref{l:update-D'}, where it is increased with a new edge $(w,A,x)$ where $w,A$ and $x$ come from the generalized item $i = \Rule{A}{\set{w\m} \dots \set{x\m, \dots}}$.
By \defref{relation}{terminal} we can conclude that line~\ref{l:let-D'} is a valid step; however, for line~\ref{l:update-D'} it depends on whether the generalized items $i \in \I$ were constructed correctly.

\begin{proposition}
\label{prop:correctness}
Algorithm~\ref{alg:shared-prefixes} computes $D'$ such that for all $(x,A) \in Q$
$$ \forall y.\ ( y \in \R{x}{A} \iff (x,A,y) \in D') $$
\begin{proof}
This follows from Propositions~\ref{prop:relations} and~\ref{prop:I}.
\end{proof}
\end{proposition}

\subsection{Time and Space Complexity}
In this section, we show the time and space complexity of our algorithm.
Our proof is based on the finite number of elements in the sets it manipulates.

\begin{proposition}[Worst-case Space Complexity] The worst-case space complexity of Algorithm~1 is $\bigO(|V|^2 \cdot |P| \cdot k)$.
\begin{proof}
The maximum size that $D'$ and \I\ may reach is:
\begin{itemize}
\item[$D'$:] The algorithm increments the graph $D'$ with non-terminal-labeled edges, so it uses at most:
\begin{eqnarray}|D'| &=& |V| \cdot |\nonterminals \cup \terminals| \cdot |V|\end{eqnarray}
what is $\bigO(|V|^2 \cdot |\nonterminals \cup \terminals|)$.

\item[\I:] The set \I\ contains generalized items, which are annotated production rules with a single vertex at the start of the right-hand side.
So we have at most:
\begin{eqnarray}
|\I| &=& |V| \cdot |P|
\end{eqnarray}
For each trace item, the number of position set sets depends on the size of the right-hand side of a production rule.
Assuming that $k$ denotes the greatest size of the right-hand side of the rules in $P$, each trace item may have $k$ position sets of size at most $|V|$ (notice that the first position set on each trace item is always a singleton).

In this context, the worst case in space complexity for $I$ is:
\[ |V| \cdot |P| \cdot k \cdot |V|.
\]
what is $\bigO(|V|^2\cdot |P| \cdot k)$.

\end{itemize}

We can now estimate the worst-case space complexity as:
\begin{eqnarray}
& \bigO(|V|^2 \cdot (|\nonterminals \cup \terminals| + |P| \cdot k)) &
\end{eqnarray}

% Considering that, for most practical applications we have to process small-sized grammars and large graphs the size of the graph dominates all the other parameters.
% In these practical cases we have that the worst-case space complexity of our algorithms is $\bigO(|V|^2)$.

\end{proof}
\end{proposition}

\begin{proposition}[Worst-case Runtime Complexity] 
\label{prop:time}
The worst-case runtime complexity of Algorithm~1 is $\bigO(|V|^3 \cdot |P|^2 \cdot k^2)$.
\begin{proof}[Proof Sketch]
The main loop iterates until there are no more unmarked vertices $x\u$.
The maximum number of unmarked vertices is given by $|I|\cdot k \cdot |V|$, where $k$ is the maximum number of possible position sets for rules of the grammar (the greatest size of a right-hand side of the rules in $P$, plus one). 
So, as $|I| = |V| \cdot |P|$, we have at most $|V|^2 \cdot |P| \cdot k$ possible vertices $x\u$.

For each iteration, the form of the trace item $i$ guides the operation to be performed.
The tests at lines~\ref{l:case-mid} and~\ref{l:case-end} have constant cost.

There are two cases to be considered inside the \textbf{switch} command:
\begin{itemize}
\item The evaluation of the condition at line~\ref{l:in-Res} requires searching over the set of trace items $I$. 
The cost of this operation is constant (supposing that we use a matrix representation).

Line~\ref{l:advance} is the case where the algorithm advances one step on a path by looking for edges $(x, \alpha, y) \in D'$.
As there are at most $|V|$ possible destination vertexes, the algorithm performs at most $|V|$ operations in this case.

At line~\ref{l:update-I}, the algorithm adds new trace items to \I\ in order to start a new derivation.
This line ensures that the algorithm only creates at most one trace item for each production rule in $P$ for a fixed vertex $x$. 
So, in this case, the algorithm performs at most $|P|$ constant time operations.

In this way, the overall cost of the case spanning from line~\ref{l:case-mid} to~\ref{l:update-I} is bounded by $\max(|V|,|P|)$. 

\item The second case of the \textit{switch} command adds non-terminal labelled edges to the graph.
The creation of such edges is performed at line~\ref{l:update-D'}, in constant time.

The appearance of a new edge triggers the update of position sets by the iteration at line~\ref{l:foreach}.
We have at most $|V|\cdot|P|\cdot k$ position sets.
Assuming, again, a matrix representation, locating each set $C$ in a trace item, requires constant time.
Thus, line~\ref{l:notify} will be executed $|V|\cdot|P|\cdot k$ times in the worst case. 

In this way, the overall cost of the case spanning from line~\ref{l:case-end} to~\ref{l:notify} is bounded by $|V|\cdot|P|\cdot k$.
\end{itemize}

This shows that the worst-case time complexity of our algorithm is $\bigO(|V|^3 \cdot |P|^2 \cdot k^2)$.

% \[\begin{array}{ll}
% \bigO (|V|^2 * |P| * k * |V|^2 * \log(V) * |P| * k) = \\
% \qquad\qquad\qquad\qquad\qquad\qquad
% \bigO(|V|^4 * \log(V) * |P|^2 * k^2).
% \end{array}
% \]
\end{proof}
\end{proposition}

\section{Related Work}
\label{sec:related-work}

Graph databases have become popular in the last few years.
Specifying queries over such databases normally include \textit{property paths}, which define paths on the data graph by means of regular expressions~\cite{mendelzon1989regular-simple-paths, w3c2012sparql-query-lang}.
In~\cite{abiteboul1995foundations,Hellings14,zhang2016,grigorev2016ll}, it is noted that there exist useful queries that cannot be expressed by regular expressions, since they require some kind of bracket matching.
\textit{Same Generation Queries}~\cite{abiteboul1995foundations} are an example of queries that cannot be expressed by regular expressions, requiring the identification of context-free paths.

Answering context-free path queries is NP-Complete~\cite{MendelzonW95}.
However, specifying the starting node of the path makes the cost of processing those queries manageable.

In~\cite{Hellings14}, the author proposes an algorithm to evaluate Context-Free Path Queries based on Earley's and CYK parsing techniques~\cite{grune2007parsing}. 
This algorithm receives a grammar (in Chomsky Normal Form) and a data graph.
The algorithm is based on the idea of adding a non-terminal-labelled edge to link nodes that are connected by a path generated by the grammar.
Regardless of the query, the algorithm in~\cite{Hellings14} processes the whole graph.
For any vertices $x$ and $y$ and non-terminal symbol $S$, an $S$-labelled edge linking $x$ to $y$ is created if there exist an $S$-derivable path in the graph linking $x$ to $y$.
After that, atomic queries can be executed in constant time.
The algorithm is $\bigO(|N||E|+(|N||V|)^3)$, where $N$ is the set of non-terminal symbols of the grammar, $V$ is the set of nodes of the graph and $E$ is the set of edges.

In~\cite{zhang2016}, the query language cfSPARQL is proposed.
The language includes queries defined by context-free grammars, as well as by nested regular expressions~\cite{nsparql}.
The evaluation mechanism of cfSPARQL is an adaptation of the algorithm in~\cite{Hellings14} and presents the same time complexity.

An LL-based approach to recognize context-free paths in RDF graphs is proposed in~\cite{grigorev2016ll}. 
The proposal uses the GLL~\cite{gll} parsing technique to define an algorithm for querying data graphs with time complexity of $\bigO( {|V|^3  max_{v \in V}(deg^+(v))})$, where $V$ is the set of vertices and $deg^+(v)$ is the outdegree of vertex $v$. 
Notice that for complete graphs this runtime complexity is $\bigO(|V|^4)$.

The Valiant's parsing algorithm~\cite{Valiant75} is the base for the query algorithm presented in~\cite{azimov-grigorev2017matrix}.
The algorithm uses a matrix representation of the graph where each cell contains the edge between two vertices, represented by line and column.
The proposal uses an efficient, GPU-based calculation of the transitive closure of that matrix to answer queries.
Similarly to~\cite{Hellings14}, the algorithm in~\cite{azimov-grigorev2017matrix} calculates all possible non-terminal labelled edges between nodes of the graph.
The time complexity of this algorithm is $\bigO(|V|^4 \cdot |N|^3)$, where $V$ is the set of vertices of the graph and $N$ is the set of non-terminal symbols of the query's grammar.

In~\cite{fred}, the authors present a Context-Free Path Query processing algorithm based on the well-known bottom-up LR parsing technique~\cite{aho2007compilers}.
The algorithm uses the LALR parsing table for the grammar.
The proposal extends Tomita's algorithm and GSS data structure~\cite{tomita} to simultaneously discover context-free paths on a data graph.
The proposed algorithm does not need to pre-process the whole graph in order to answer the query.
The time complexity of this algorithm is given by $\bigO(|V|^{4+k} \cdot |I|^{1+k} \cdot |\terminals| \cdot |\nonterminals|)$, where $k$ is the maximum size of the right-hand side of the production rules in the grammar and $I$ is the number of lines of the LALR(1) parsing table.

In~\cite{MEDEIROS201975}, the authors  propose a query processing algorithm based on the LL parsing technique~\cite{aho2007compilers}.
For queries of the form $(x, S)$, where $x$ is a vertex of the graph and $S$ is a non-terminal symbol, the algorithm proceeds in a top-down manner, trying to discover $S$-generated paths from $x$.
The worst case runtime complexity of their algorithm is $\bigO(|V|^3 \cdot |P|)$, where $P$ is the set of production rules of the grammar.

The authors in~\cite{jochemkuijpers} evaluate the Context-Free Path Query evaluation methods in~\cite{azimov-grigorev2017matrix,fred,Hellings2015pathresults}.
The authors perform experiments with several data sets, including real and synthetic ones.
The paper focus on scalability of the three approaches and concludes that these methods are not yet adequate for big data processing.
We expect to contribute towards that goal.

\bigskip

\section{Experiments}
In this section we present some performance experiments to investigate the viability of our algorithm.
We implemented a prototype using the Go programming language\footnote{The source code and data for out prototype is available at Github; the link to it is not shown due to the double-blind revision process of the conference.}.
The experiments were performed on a Debian 8.11, 64GB RAM,  Intel Xeon E312xx (Sandy Bridge) @ 2.195GHz, 64 bits.
The results presented here are the average time and memory of 10 runs.

We compared our algorithm to the one in~\cite{comlan}.
Their algorithm is implemented in Python and was run using the same computer as the algorithm we propose here. 
For both algorithms, we performed the same experiments as in~\cite{grigorev2016ll, azimov-grigorev2017matrix, Hellings2015pathresults, comlan, zhang2016, jochemkuijpers}.
The databases used in the experiments include both synthetic graphs and publicly available ontologies.
The synthetic graphs and the grammars used to query them were designed in order to explore specific characteristics of the evaluation mechanisms, such as their memory and runtime performance in their worst-case or random scenarios; the influence of grammar ambiguity or density/sparsity as well as to observe the scalability properties of our approach.
The  dataset of ontologies consists of a number of popular ontologies publicly available and it is the same used in previous works~\cite{grigorev2016ll, azimov-grigorev2017matrix, zhang2016}.

The non-random synthetic graphs used in the experiments are described as follows.
A complete graph corresponds to the product $V \times \terminals \times V$, and it represents the worst-case scenario for the database, where each vertex is linked to all the vertices of the graph, including itself.
We also considered two kinds of linear graphs, \textit{i.e.}, graphs that have the form of a single straight path:
the first kind, referred to as $ab$-list graphs, is formed by graphs whose labels form a path $a^n b^n$; the second kind, called $\sigma$-string graphs, is formed by straight line graphs where all the edges are labeled with $\sigma$.
Cycle graphs have all edges labeled with $\sigma$.

\bigskip

Let us present some experiments to test the behaviour of our algorithm in specific cases.

\paragraph{Dealing with Ambiguous Grammars}

The data presented in Figure~\ref{fig:exp-string} %and~\ref{fig:exp-complete}
corresponds to the execution over $ab$-list. % and complete graphs, respectively.
We used Grammars~\ref{gram:ab_ambiguous} and~\ref{gram:ab_unambiguous}, which recognize the language of balanced $a$'s and $b$'s.
These grammars are defined as follows:

\begin{grammar}(Ambiguous) Generates strings containing balanced pairs of $a$'s and $b$'s~\cite{grigorev2016ll, azimov-grigorev2017matrix, zhang2016, comlan}:
\label{gram:ab_ambiguous}
$$\Rule{S}{S~S ~|~ a ~S~ b ~|~ \epsilon}$$
\end{grammar}

\begin{grammar} Unambiguous grammar generating the same language of~\autoref{gram:ab_ambiguous}~\cite{grigorev2016ll,azimov-grigorev2017matrix,zhang2016,comlan}:
\label{gram:ab_unambiguous}
$$\Rule{S}{a ~S ~b ~S ~|~ \epsilon}$$
\end{grammar}

The query was defined as $Q=\set{(x, S)\ |\ x\in V}$ \textit{i.e.}, we look for all vertices that are linked by an $S$-derived path from each vertex of the $ab$-list graph. 

%Both for $ab$-list and complete graphs,
We observe that our algorithm presents a very efficient runtime behaviour as the graph grows in size, when compared to~\cite{comlan}.
We also observe that the behaviour of our algorithm is not heavily affected by the grammar's ambiguity.

In terms of memory consumption, both algorithms behave in a similar way, with a small advantage to our algorithm.

%This shows that, depending on the circumstances, grammar ambiguity is not necessarily prejudicial for the algorithm's performance.

\paragraph{Dense and Sparse Grammars.}
Figures~\ref{fig:exp-hlg-cycle} and~\ref{fig:exp-hlg-path} % and~\ref{fig:exp-hlg-complete}
compare the execution of our prototype and the LL~\cite{comlan} algorithm over cycle and path
%and complete
graphs,
respectively, using Grammars~\ref{gram:hlg_dense} and~\ref{gram:hlg_sparse} and for the same query set as before.

\begin{grammar} Dense grammar recognizing the language $\sigma^+$~\cite{Hellings2015pathresults}:
\label{gram:hlg_dense}
$$\Rule{A}{A ~A} \qquad \Rule{A}{\sigma}$$
\end{grammar}

The notion of a \textit{dense} grammar refers to the fact of the grammar generating strings without having empty transitions, in contrast to a \textit{sparse} grammar.

\begin{grammar} Sparse grammar recognizing the language $\sigma^*$~\cite{Hellings2015pathresults}:
\label{gram:hlg_sparse}
$$\Rule{B}{B ~A ~|~ A ~B ~|~ \epsilon} \qquad \Rule{A}{\sigma}$$
\end{grammar}

As in the previous case, we observe that the behaviour of our algorithm is better in terms of time and memory consumption, when compared to the algorithm in~\cite{comlan}.

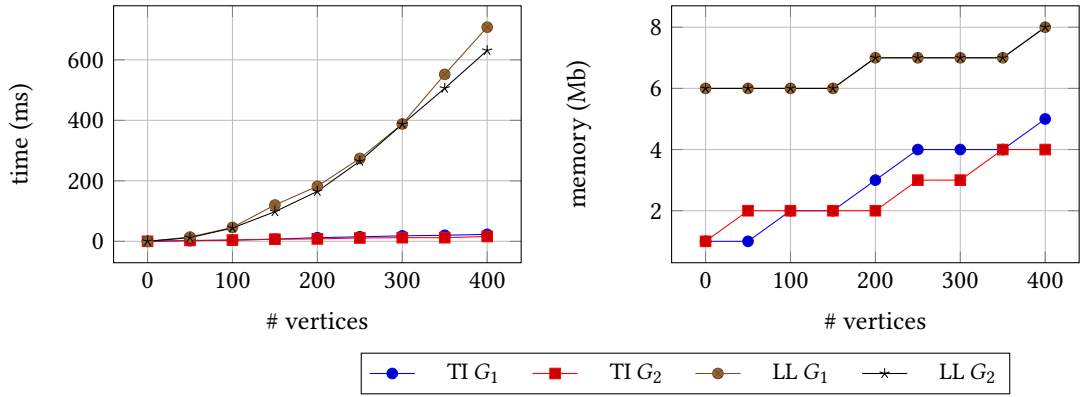
\begin{figure*}
    \centering
    \begin{tikzpicture}
    \begin{axis}[
        name=ax2,
        height=5cm,
        width=7cm,
        grid=major,
        xlabel=\# vertices,
        ylabel=time (ms),
        legend style={font=\small,
        at={(2.2,-0.35)},
        anchor=north east,legend columns=4, column sep=10pt}
    ]

\addplot coordinates { (0, 0) (50, 2) (100, 4) (150, 7) (200, 12) (250, 15) (300, 18) (350, 20) (400, 23)}; % (450, 47) (500, 45)};
\addlegendentry{TI $G_{\ref{gram:ab_ambiguous}}$};
\addplot coordinates { (0, 0) (50, 3) (100, 4) (150, 7) (200, 8) (250, 11) (300, 12) (350, 12) (400, 16)}; % (450, 13) (500, 13)};
\addlegendentry{TI $G_{\ref{gram:ab_unambiguous}}$};

\addplot coordinates { (0, 0) (50, 14) (100, 46) (150, 120) (200, 182) (250, 274) (300, 388) (350, 552) (400, 708)};
\addlegendentry{LL~$G_{\ref{gram:ab_ambiguous}}$};
\addplot coordinates { (0, 0) (50, 12) (100, 44) (150, 98) (200, 165) (250, 265) (300, 387) (350, 506) (400, 631)};
\addlegendentry{LL~$G_{\ref{gram:ab_unambiguous}}$};

\end{axis}

\begin{axis}[
        at={(ax2.south east)},
        xshift=2cm,
        height=5cm,
        width=7cm,
        grid=major,
        legend pos=north west,
        xlabel=\# vertices,
        ylabel=memory (Mb)
    ]

\addplot coordinates { (0, 1) (50, 1) (100, 2) (150, 2) (200, 3) (250, 4) (300, 4) (350, 4) (400, 5)}; % (450, 7) (500, 7)};
%\addlegendentry{This work $G_{\ref{gram:ab_ambiguous}}$};
\addplot coordinates { (0, 1) (50, 2) (100, 2) (150, 2) (200, 2) (250, 3) (300, 3) (350, 4) (400, 4)}; % (450, 5) (500, 4)};
%\addlegendentry{This work $G_{\ref{gram:ab_unambiguous}}$};
\addplot coordinates { (0, 6) (50, 6) (100, 6) (150, 6) (200, 7) (250, 7) (300, 7) (350, 7) (400, 8)};
%\addlegendentry{ LL $G_{\ref{gram:ab_ambiguous}}$};
\addplot coordinates { (0, 6) (50, 6) (100, 6) (150, 6) (200, 7) (250, 7) (300, 7) (350, 7) (400, 8)};
%\addlegendentry{ LL $G_{\ref{gram:ab_unambiguous}}$};

\end{axis}
\end{tikzpicture}
\caption{$ab$-list graphs, Grammars~$G_{\ref{gram:ab_ambiguous}}$ and~$G_{\ref{gram:ab_unambiguous}}$.}
    \label{fig:exp-string}
\end{figure*}

%\begin{figure*}
%    \centering
%\input{charts/ab-complete}
%\caption{Complete graphs, Grammars~$G_{\ref{gram:ab_ambiguous}}$ and~$G_{\ref{gram:ab_unambiguous}}$.}
%    \label{fig:exp-complete}
%\end{figure*}

% \begin{figure*}[htpb]
%     \centering
%     \begin{subfigure}[t]{0.45\textwidth}
%         \centering
%         \input{charts/string}
%         \caption{String Graphs with Grammars G\textsubscript{\ref{gram:ab_ambiguous}} and G\textsubscript{\ref{gram:ab_unambiguous}}.\ciro{atualizado 21oct 16:11}}
%         \label{fig:exp-string}
%     \end{subfigure}
%     ~
%     \begin{subfigure}[t]{0.45\textwidth}
%         \centering
%         \input{charts/complete}
%         \caption{Complete Graphs with Grammars G\textsubscript{\ref{gram:ab_ambiguous}} and G\textsubscript{\ref{gram:ab_unambiguous}}.}
% \label{fig:exp-complete}
%     \end{subfigure}
%     \caption{Runtime (ms) and Memory (Mb) Performance}
%     \label{fig:experiments}
% \end{figure*}

\begin{figure*}
    \centering
    \begin{tikzpicture}
    \begin{axis}[
        name=ax2,
        height=5cm,
        width=7cm,
        grid=major,
        xlabel=\# vertices,
        ylabel=time (ms),
        ymax=10000,
        legend style={font=\small,
        at={(2.2,-0.35)},
        anchor=north east,legend columns=4, column sep=10pt}
    ]

\addplot coordinates { (0, 0) (50, 37) (100, 262) (150, 811) (200, 1474) (250, 2656) (300, 3991) (350, 5622) (400, 7898)}; % (450, 12108) (500, 14764)};
\addlegendentry{TI G\textsubscript{\ref{gram:hlg_dense}}};

\addplot coordinates { (0, 0) (50, 16) (100, 103) (150, 292) (200, 489) (250, 862) (300, 1062) (350, 1378) (400, 1700)}; % (450, 2352) (500, 2946)};
\addlegendentry{TI G\textsubscript{\ref{gram:hlg_sparse}}};

\addplot coordinates { (0, 0) (50, 19676)};
\addlegendentry{LL~$G_{\ref{gram:hlg_dense}}$ };

\addplot coordinates { (0, 0) (50, 1065)};
\addlegendentry{LL~$G_{\ref{gram:hlg_sparse}}$ };

\end{axis}

\begin{axis}[
        at={(ax2.south east)},
        xshift=2cm,
        height=5cm,
        width=7cm,
        grid=major,
        legend pos=north west,
        xlabel=\# vertices,
        ylabel=memory (Mb)
    ]

\addplot coordinates { (0, 1) (50, 3) (100, 8) (150, 15) (200, 19) (250, 28) (300, 31) (350, 36) (400, 42)}; % (450, 84) (500, 72)};
%\addlegendentry{TI G\textsubscript{\ref{gram:hlg_dense}}};

\addplot coordinates { (0, 1) (50, 3) (100, 8) (150, 14) (200, 19) (250, 29) (300, 39) (350, 41) (400, 47)}; % (450, 85) (500, 83)};
%\addlegendentry{TI G\textsubscript{\ref{gram:hlg_sparse}}};

\addplot coordinates { (0, 6) (50, 6)};
%\addlegendentry{LL~$G_{\ref{gram:hlg_dense}}$ };

\addplot coordinates { (0, 6) (50, 6)};
%\addlegendentry{LL~$G_{\ref{gram:hlg_sparse}}$ };

\end{axis}
\end{tikzpicture}
    \caption{Cycle graphs, Grammars G\textsubscript{\ref{gram:hlg_dense}} and G\textsubscript{\ref{gram:hlg_sparse}}.}
    \label{fig:exp-hlg-cycle}
\end{figure*}
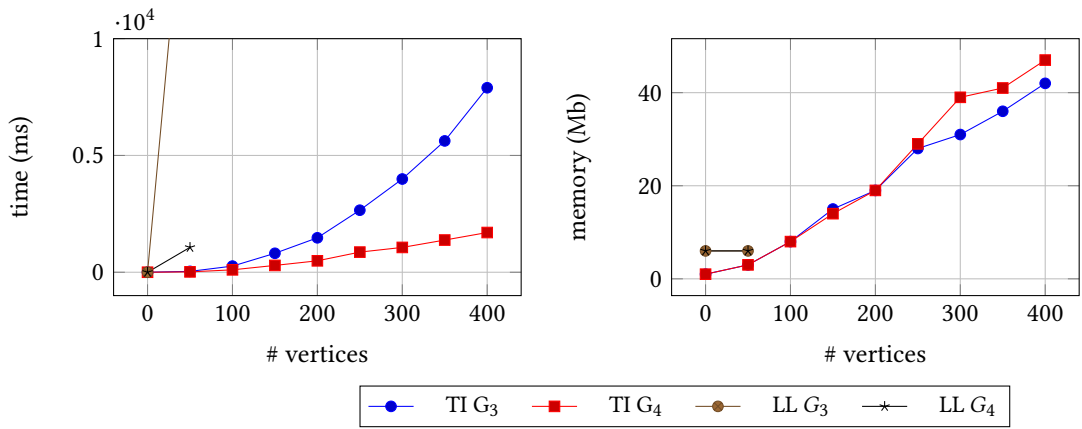

\begin{figure*}[htpb]
        \centering
        \begin{tikzpicture}
    \begin{axis}[
        name=ax2,
        height=5cm,
        width=7cm,
        grid=major,
        xlabel=\# vertices,
        ylabel=time (ms),
        ymax=20000,
        legend style={font=\small,
        at={(2.2,-0.35)},
        anchor=north east,legend columns=4, column sep=10pt}
    ]

\addplot coordinates { (0, 0) (50, 11) (100, 71) (150, 220) (200, 466) (250, 743) (300, 1152) (350, 1605) (400, 2155)}; % (450, 2799) (500, 3685)};
\addlegendentry{TI G\textsubscript{\ref{gram:hlg_dense}}};
\addplot coordinates { (0, 0) (50, 7) (100, 45) (150, 123) (200, 276) (250, 416) (300, 631) (350, 850) (400, 1036)}; % (450, 1193) (500, 1536)};
\addlegendentry{TI G\textsubscript{\ref{gram:hlg_sparse}}};

\addplot coordinates { (0, 0) (50, 17465) (100, 517852)}; % (150, 3893564)};
\addlegendentry{LL G\textsubscript{\ref{gram:hlg_dense}}};
\addplot coordinates { (0, 0) (50, 325) (100, 3720) (150, 15822) (200, 46913)};
\addlegendentry{LL G\textsubscript{\ref{gram:hlg_sparse}}};

\end{axis}

\begin{axis}[
        at={(ax2.south east)},
        xshift=2cm,
        height=5cm,
        width=7cm,
        grid=major,
        legend pos=north west,
        xlabel=\# vertices,
        ylabel=memory (Mb)
    ]

\addplot coordinates { (0, 1) (50, 2) (100, 6) (150, 10) (200, 13) (250, 18) (300, 21) (350, 26) (400, 28)}; % (450, 36) (500, 44)};
%\addlegendentry{TI G\textsubscript{\ref{gram:hlg_dense}}};
\addplot coordinates { (0, 1) (50, 2) (100, 7) (150, 11) (200, 15) (250, 20) (300, 25) (350, 29) (400, 35)}; % (450, 43) (500, 50)};
%\addlegendentry{TI-based G\textsubscript{\ref{gram:hlg_sparse}}};

\addplot coordinates { (0, 6) (50, 6) (100, 9) (150, 10)};
%\addlegendentry{LL G\textsubscript{\ref{gram:hlg_dense}}};
\addplot coordinates { (0, 6) (50, 6) (100, 7) (150, 8) (200, 9)};
%\addlegendentry{LL G\textsubscript{\ref{gram:hlg_sparse}}};

\end{axis}
\end{tikzpicture}
    \caption{$\sigma$-string graphs, Grammars G\textsubscript{\ref{gram:hlg_dense}} and G\textsubscript{\ref{gram:hlg_sparse}}.}
    \label{fig:exp-hlg-path}
\end{figure*}
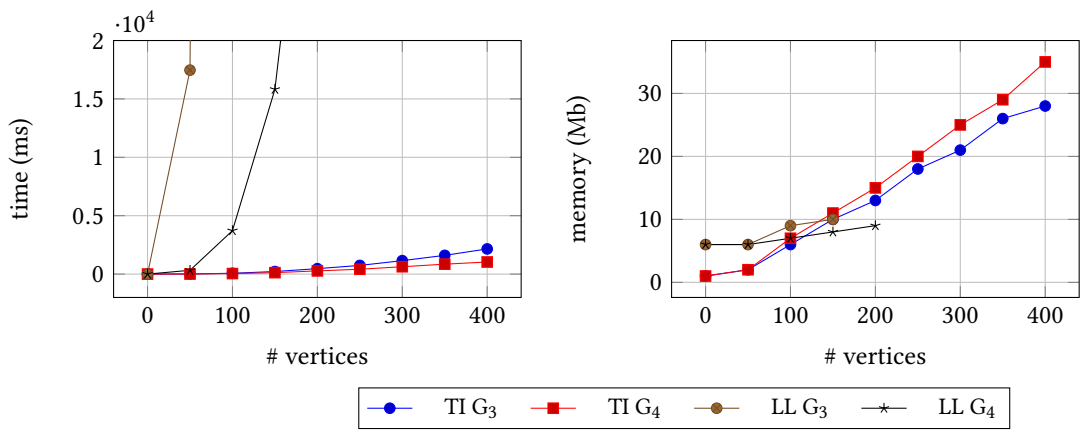

%\begin{figure*}[htpb]
%    \centering
%        %\input{charts/hellings-complete.tex}
%    \caption{Complete Graphs, Grammars G\textsubscript{\ref{gram:hlg_dense}} and G\textsubscript{\ref{gram:hlg_sparse}}.}
%    \label{fig:exp-hlg-complete}
%\end{figure*}

For all graphs used in this experiment, our prototype presented a time performance that seems to be better than the one given by Proposition~\ref{prop:time}.

Notice that the form of the grammar's production rules have an important influence over the time performance of the algorithms. 
For $\sigma$-string and cycle graphs, sparse grammars seem to have an advantage over dense grammars.
%For complete graphs, the use of dense grammars seems to be preferable.

% It is important to notice that this is due not only to the shape of the graph, but also to the grammar.
% For cycle graphs (Figure~\ref{fig:exp-hellings-cycle}), there exists a valid path between every pair of vertices in the graph, what leads to a quadratic number of results.
% For path graphs (Figure~\ref{fig:exp-hellings-path}), the number of results grows quadratically with the size of the graph, and so does the performance curve.
% This is due to the fact that the non-terminal-labeled edges computed by the algorithm are used as input to itself.
% For complete graphs Figure~\ref{fig:exp-hellings-complete}, we have a similar case to cycle graphs, where there exists a path between every pair of vertices.
% \ciro{improvements over state-of-art?}

Regarding memory consumption, we observe the same situation as for the previous case, with our algorithm performing slightly better than the one in~\cite{comlan}.

\paragraph{Experiment with ontologies.}
For the next experiment we used a set of popular ontologies publicly available on the internet.
This dataset and the grammars described below are the same used in previous works~\cite{grigorev2016ll, azimov-grigorev2017matrix, zhang2016, comlan}.
The ``geospecies'' database and Grammar~\ref{gram:bt} were used in~\cite{jochemkuijpers}.

Grammar~\ref{gram:sc_t} retrieves concepts in the same level of the RDFS' $\mathit{subClassOf}/\mathit{type}$ hierarchy.
The experiment consists on performing a ``same generation query''~\cite{abiteboul1995foundations}.
For each vertex of the graph, the query looks for all vertices that are at the same level in the graph of the subclass/type hierarchy.

\begin{grammar} Retrieves concepts in the same level of hierarchy~\cite{grigorev2016ll, azimov-grigorev2017matrix, zhang2016, comlan}:
\label{gram:sc_t}
$$\Rule{S}{\textit{subClassOf}\ ~S~ \textit{subClassOf}^{-1}} \qquad \Rule{S}{\textit{type} ~S~ \textit{type}^{-1}}$$
$$\Rule{S}{\textit{subClassOf}~ \textit{subClassOf}^{-1}}
 \qquad\Rule{S}{\textit{type} ~\textit{type}^{-1}}$$
\end{grammar}

Grammar~\ref{gram:sc} retrieves concepts in adjacent levels of the RDFS' $\mathit{subClassOf}$ hierarchy.
\begin{grammar} Retrieves concepts on adjacent levels of the hierarchy of classes in RDF~\cite{grigorev2016ll, azimov-grigorev2017matrix, zhang2016, comlan}:
\label{gram:sc}
$$\Rule{S}{B ~\textit{subClassOf}^{-1}} \qquad \Rule{B}{\textit{subClassOf}\ ~ B~ \textit{subClassOf}^{-1} ~|~ \epsilon}$$
\end{grammar}

Grammar~\ref{gram:bt} retrieves concepts in the same level of the $\mathit{broaderTransitive}$ hierarchy.
These edges are directed from child to parent, relating  categories of species, families, orders, etc.
This is a real example of application, where a Context-Free Path Query is used to identify the pairs of vertices that are in the same category inside the biological taxonomy.

\begin{grammar} Retrieves concepts on adjacent levels of hierarchy~\cite{jochemkuijpers}:
\label{gram:bt}
$$\Rule{S}{\textit{broaderTransitive} ~S~ \textit{broaderTransitive}^{-1}}$$
$$\Rule{S}{\textit{broaderTransitive} ~ \textit{broaderTransitive}^{-1}}$$
\end{grammar}

The results of running our algorithm (as well as LL~\cite{comlan}) are shown in Table~\ref{tab:ontologies}.
The query used in this case was the same as in the previous cases: we look for paths departing from each vertex of the graph.
The first three columns of the table show the used grammar and ontology, the size of the graph and the number of results obtained by the query.

In the data presented in Table~\ref{tab:ontologies}, we can observe that both algorithms behave in the same way as observed for the synthetic examples given previously.
In general, our algorithm performs better that the one in~\cite{comlan}, with a great difference in time, in favor to our algorithm.
The last line in Table~\ref{tab:ontologies} does not contain data for the LL algorithm, since our computational resources were not sufficient for the normal execution of that algorithm.

%\mm{Eu vou parar aqui e vou passar para as conclusoes. @Ciro: Caso necessario, a gente pode deixar os experimentos que faltam para a proxima versao... Eu nao estou confiante em que este paper vai ser aceito... mas precisamos submeter.}

% Headers of Tables~\ref{tab:ontologies},~\ref{tab:kuijpers-anbmcmdn},~\ref{tab:kuijpers-sparse} and~\ref{tab:kuijpers-ambiguous}:
% \begin{itemize}
%     \item Database and Grammar are self-explanatory.
%     \item Results is the number of edges in $D'$ that are labeled with the grammar's start symbol, i. e., $|(\_,S,\_)|$.
%     \item Time is the time taken by the algorithm to prepare the data structures, process the query and return the results.
%     \item Memory was obtained with Go's \texttt{runtime} package.
%     It's the sum of stack, heap and system in use.
% \end{itemize}

\begin{table*}[]
\centering\small
\begin{tabular}{@{}lrr|rr|rr@{}}
\toprule
&&& \multicolumn{2}{c|}{\textbf{This work}} & \multicolumn{2}{c}{\textbf{LL}~\cite{comlan}} \\
\textbf{Grammar \& Graph} & \textbf{$|V|$} & \textbf{Results} & \textbf{Time} & \textbf{Memory} & \textbf{Time} & \textbf{Memory}\\
\midrule

$G_{\ref{gram:sc_t}}$, skos & 43 & 810 & 4 ms & 2.5 Mb & 115 ms & 6.7 Mb \\
$G_{\ref{gram:sc_t}}$, generations & 82 & 2164 & 8 ms & 2.9 Mb & 411 ms & 7.3 Mb \\
$G_{\ref{gram:sc_t}}$, travel & 92 & 2499 & 10 ms & 3.8 Mb & 1139 ms & 7.4 Mb \\
$G_{\ref{gram:sc_t}}$, univ\_bench & 90 & 2540 & 9 ms & 3.9 Mb & 1226 ms & 7.4 Mb \\
$G_{\ref{gram:sc_t}}$, foaf & 93 & 4118 & 15 ms & 4.6 Mb & 1915 ms & 7.4 Mb \\
$G_{\ref{gram:sc_t}}$, people\_pets & 163 & 9472 & 48 ms & 7.0 Mb & 7614 ms & 9.7 Mb \\
$G_{\ref{gram:sc_t}}$, funding & 272 & 17634 & 151 ms & 12.1 Mb & 32059 ms & 11.6 Mb \\
$G_{\ref{gram:sc_t}}$, atom\_primitive & 142 & 15454 & 208 ms & 14.6 Mb & 48048 ms & 11.0 Mb \\
$G_{\ref{gram:sc_t}}$, biomedical & 134 & 15156 & 165 ms & 12.7 Mb & 43248 ms & 11.4 Mb \\
$G_{\ref{gram:sc_t}}$, pizza & 359 & 56195 & 407 ms & 18.5 Mb & 371402 ms & 19.2 Mb \\
$G_{\ref{gram:sc_t}}$, wine & 468 & 66572 & 425 ms & 20.9 Mb & 389951 ms & 21.5 Mb \\
\midrule

$G_{\ref{gram:sc}}$, skos & 2 & 1 & 0 ms & 1.5 Mb & 0 ms & 6.6 Mb \\
$G_{\ref{gram:sc}}$, generations & 0 & 0 & 0 ms & 1.6 Mb & 0 ms & 6.6 Mb \\
$G_{\ref{gram:sc}}$, travel & 32 & 63 & 3 ms & 2.0 Mb & 9 ms & 6.6 Mb \\
$G_{\ref{gram:sc}}$, univ\_bench & 42 & 81 & 4 ms & 2.2 Mb & 11 ms & 6.7 Mb \\
$G_{\ref{gram:sc}}$, foaf & 13 & 10 & 0 ms & 2.4 Mb & 0 ms & 6.6 Mb \\
$G_{\ref{gram:sc}}$, people\_pets & 44 & 37 & 4 ms & 3.0 Mb & 5 ms & 6.6 Mb\\
$G_{\ref{gram:sc}}$, funding & 93 & 1158 & 12 ms & 3.5 Mb & 932 ms & 7.4 Mb\\
$G_{\ref{gram:sc}}$, atom\_primitive & 124 & 122 & 86 ms & 11.5 Mb & 20 ms & 6.7 Mb\\
$G_{\ref{gram:sc}}$, biomedical & 123 & 2871 & 34 ms & 6.3 Mb & 3211 ms & 8.0 Mb\\
$G_{\ref{gram:sc}}$, pizza & 261 & 1262 & 34 ms & 6.9 Mb & 2019 ms & 7.9 Mb \\
$G_{\ref{gram:sc}}$, wine & 163 & 133 & 7 ms & 3.5 Mb & 58 ms & 7.0 Mb \\
\midrule

$G_{\ref{gram:bt}}$, geospecies & 20882 & 226669749 & 624352 ms & 36844.7 Mb & N/A & N/A  \\
\bottomrule
\end{tabular}
\caption{Performance Evaluation on RDF Databases.}
\label{tab:ontologies}
\end{table*}

\paragraph{Querying Random graphs.}
The next experiments were proposed by~\cite{jochemkuijpers} and use random, synthetic graphs.
We used a graph generator function based on the definition given by by~\cite{albertbarabasi}.
Given the size of the graph in number of vertices $n$ and a constant $k \leq n$, the generator function, denoted by $\Gnk(n,k)$, starts with a clique of $k$ vertices.
For each $v$ in the $n-k$ remaining vertices, the generator adds $k$ edges from $v$ to any vertices already in the graph.
The edge labels are randomly chosen, being either $a,b,c$ or $d$.
The probability for a vertex to be chosen is directly proportional to its degree at that moment, such that the higher the degree of the vertex, higher is its probability receive the new edges.
%As we did not have access to the same databases used in their paper, we implemented our own random graph generator.

\begin{grammar} Defines the language $a^nb^mc^md^n$~\cite{jochemkuijpers}:
\label{gram:kjp_an_bm_cm_dn}
$$\Rule{S}{a ~S ~d ~|~ a ~X ~d} \qquad
\Rule{X}{b ~X ~c} ~|~ \epsilon$$
\end{grammar}

%\begin{grammar} Defines the language $a^+$~\cite{jochemkuijpers}:
%\label{gram:kjp_sl}
%$$\Rule{S_L}{a ~S_L ~|~ \epsilon}$$
%\end{grammar}

%\begin{grammar} Defines the language $a^+$~\cite{jochemkuijpers}:
%\label{gram:kjp_sr}
%$$\Rule{S_R}{a ~S_R ~|~ \epsilon}$$
%\end{grammar}

%\begin{grammar} Defines the language $a^n(b|c|d)^{4}a^{-1n}$, $n \geq 2, m \geq 1$~\cite{jochemkuijpers}:
%\label{gram:kjp_sparse}
%$$\Rule{S}{a ~P ~a^{-1}} \qquad
%\Rule{P}{a ~P ~a^{-1}} \qquad
%\Rule{P}{a ~M ~a^{-1}} $$
%$$ \Rule{M}{X ~X ~X ~X} \qquad
%\Rule{X}{b ~|~ c ~|~ d}$$
%\end{grammar}

%\begin{grammar} Defines the language %$(a|b|c|d)^+$~\cite{jochemkuijpers}:
%\label{gram:kjp_ambiguous}
%$$\Rule{S}{S ~S} \qquad
%\Rule{S}{a ~|~ b ~|~ c ~|~ d}$$
%\end{grammar}

\begin{table*}[htpb]
\centering\small
\begin{tabular}{@{}lrr|rr|rr@{}}
\toprule
&&& \multicolumn{2}{c|}{\textbf{This work}} & \multicolumn{2}{c}{\textbf{LL}~\cite{comlan}} \\
\textbf{Grammar \& Graph} & \textbf{$|V|$} & \textbf{Results} & \textbf{Time} & \textbf{Memory} & \textbf{Time} & \textbf{Memory}\\
\midrule

$G_{\ref{gram:kjp_an_bm_cm_dn}}$, \Gnk(100,1) & 100 & 5 & 3 ms & 2.5 Mb & 4 ms & 6.6 Mb \\
$G_{\ref{gram:kjp_an_bm_cm_dn}}$, \Gnk(500,1) & 500 & 25 & 11 ms & 4.8 Mb & 71 ms & 7.5 Mb \\
$G_{\ref{gram:kjp_an_bm_cm_dn}}$, \Gnk(2500,1) & 2500 & 161 & 143 ms & 18.0 Mb & 1372 ms & 11.7 Mb \\
$G_{\ref{gram:kjp_an_bm_cm_dn}}$, \Gnk(10000,1) & 10000 & 706 & 708 ms & 47.2 Mb & 20834 ms & 27.7 Mb \\
$G_{\ref{gram:kjp_an_bm_cm_dn}}$, \Gnk(100,3) & 100 & 56 & 6 ms & 2.3 Mb & 46 ms & 7.0 Mb \\
$G_{\ref{gram:kjp_an_bm_cm_dn}}$, \Gnk(500,3) & 500 & 769 & 44 ms & 6.9 Mb & 1954 ms & 8.6 Mb \\
$G_{\ref{gram:kjp_an_bm_cm_dn}}$, \Gnk(2500,3) & 2500 & 3377 & 232 ms & 22.0 Mb & 46668 ms & 17.7 Mb \\
$G_{\ref{gram:kjp_an_bm_cm_dn}}$, \Gnk(10000,3) & 10000 & 14583 & 1181 ms & 72.1 Mb & 826796 ms & 51.7 Mb \\
$G_{\ref{gram:kjp_an_bm_cm_dn}}$, \Gnk(100,5) & 100 & 312 & 6 ms & 2.9 Mb & 465 ms & 7.1 Mb \\
$G_{\ref{gram:kjp_an_bm_cm_dn}}$, \Gnk(500,5) & 500 & 2207 & 63 ms & 8.7 Mb & 11413 ms & 10.0 Mb \\
$G_{\ref{gram:kjp_an_bm_cm_dn}}$, \Gnk(2500,5) & 2500 & 13823 & 456 ms & 29.0 Mb & 415582 ms & 25.8 Mb \\
$G_{\ref{gram:kjp_an_bm_cm_dn}}$, \Gnk(10000,5) & 10000 & 77423 & 1946 ms & 102.9 Mb & N/A & N/A \\
$G_{\ref{gram:kjp_an_bm_cm_dn}}$, \Gnk(100,10) & 100 & 1068 & 18 ms & 3.6 Mb & 6362 ms & 7.8 Mb \\
$G_{\ref{gram:kjp_an_bm_cm_dn}}$, \Gnk(500,10) & 500 & 10211 & 249 ms & 15.4 Mb & 217209 ms & 14.8 Mb \\
$G_{\ref{gram:kjp_an_bm_cm_dn}}$, \Gnk(2500,10) & 2500 & 102867 & 1736 ms & 65.1 Mb & N/A & N/A \\
$G_{\ref{gram:kjp_an_bm_cm_dn}}$, \Gnk(10000,10) & 10000 & 784055 & 10476 ms & 350.7 Mb & N/A & N/A \\
\bottomrule
\end{tabular}
\caption{Experiment with grammar $G_{\ref{gram:kjp_an_bm_cm_dn}}$}
\label{tab:kjp-anbmcmdn}
\end{table*}

%\begin{table*}[htpb]
%\centering\small
%\input{tables/kjp-sparse}
%\caption{Experiment with grammar $G_{\ref{gram:kjp_sparse}}$}
%\label{tab:kuijpers-sparse}
%\end{table*}

%\begin{table*}[htpb]
%\input{tables/kjp-ambiguous}
%\caption{Experiment with grammar %$G_{\ref{gram:kjp_ambiguous}}$}
%\label{tab:kjp-ambiguous}
%\end{table*}

%\begin{table}[htpb]
%\centering\small
%\input{tables/kjp-sl-sr}
%\caption{Experiment with grammars $G_{\ref{gram:kjp_sl}}$ and $G_{\ref{gram:kjp_sr}}$ \ciro{não achei o tamanho do grafo, então rodei com (100,5)}}
%\label{tab:kuijpers-sl-sr}
%\end{table}

The runtimes and memory usage observed in this experiment follow the pattern of the previous ones:
our algorithm outperforms the running time observed for the LL-based algorithm, at the same time that it uses less memory.

\section{Final Remarks}
\label{sec:final}
We presented an algorithm for the evaluation of Context-Free Path Queries for RDF databases.
Our algorithm combines characteristics of previously proposed  techniques, in order to obtain better scalability.

We presented analysis about the correctness of our algorithm, as well as an estimation of its worst-case time and space complexity.

We validated our work by using both synthetic and real-life examples, showing that our prototype outperforms another, recently published algorithm.

The query processed by our algorithm may be defined to contain any context-free grammar.
Our results show that there is no significant difference in the performance of the algorithm in relation to conditions of the grammars, like ambiguity or spareness.

The practical use of our algorithm may be allowed by including it as part of a query language engine, as it is mentioned in~\cite{comlan}.

As future work, we will investigate the construction of a parallel version of our algorithm.
This may improve it's performance, since the treatment of unvisited vertices in position sets may be done in parallel.

We are also working on benchmarking protocols for algorithms for evaluation the of Context-Free Path Queries.
This would make possible to have more accurate data, in order to compare the different algorithms that are being proposed to implement this kind of queries.

% The general runtime complexity is cubic, except for the matrix-based approach~\cite{azimov-grigorev2017matrix}, which depends on the number of operations necessary for the multiplication of boolean matrices, what may be less than cubic.
% Most approaches achieve the optimal space complexity, since there might be up to $\bigO(|V|^2)$ edges in a graph.
% Our approach achieves both state-of-the-art runtime and space complexities.
% It is important to notice that these approaches allow the use of data structures that grow in size as needed, except for the matrix-based approach, which needs to allocate the matrices \ciro{verificar a veracidade dessa informação no artigo dos russos}.

%\chapter{teste}

%\section{Acknowledgements} \ciro{Isso é proibido no blind-reviewing.}
%This work is partly supported by INES grant CNPq/465614/2014-0
%and CAPES-DS.

\bibliographystyle{ACM-Reference-Format}
\bibliography{references}


\begin{thebibliography}{00}

%%% the Web address.

\ifx \showCODEN    \undefined \def \showCODEN     #1{\unskip}     \fi
\ifx \showDOI      \undefined \def \showDOI       #1{#1}\fi
\ifx \showISBNx    \undefined \def \showISBNx     #1{\unskip}     \fi
\ifx \showISBNxiii \undefined \def \showISBNxiii  #1{\unskip}     \fi
\ifx \showISSN     \undefined \def \showISSN      #1{\unskip}     \fi
\ifx \showLCCN     \undefined \def \showLCCN      #1{\unskip}     \fi
\ifx \shownote     \undefined \def \shownote      #1{#1}          \fi
\ifx \showarticletitle \undefined \def \showarticletitle #1{#1}   \fi
\ifx \showURL      \undefined \def \showURL       {\relax}        \fi
% invisible to TeX
\providecommand\bibfield[2]{#2}
\providecommand\bibinfo[2]{#2}
\providecommand\natexlab[1]{#1}
\providecommand\showeprint[2][]{arXiv:#2}

\bibitem[\protect\citeauthoryear{Abiteboul, Hull, and Vianu}{Abiteboul
  et~al\mbox{.}}{1995}]%
        {abiteboul1995foundations}
\bibfield{author}{\bibinfo{person}{S. Abiteboul}, \bibinfo{person}{R. Hull},
  {and} \bibinfo{person}{V. Vianu}.} \bibinfo{year}{1995}\natexlab{}.
\newblock \bibinfo{booktitle}{{\em Foundations of Databases}}.
\newblock \bibinfo{publisher}{Addison-Wesley}.
\newblock
\showISBNx{9780201537710}
\showLCCN{94019295}
\showURL{%
\url{https://books.google.com.br/books?id=HN9QAAAAMAAJ}}


\bibitem[\protect\citeauthoryear{Aho, Lam, Sethi, and Ullman}{Aho
  et~al\mbox{.}}{2007}]%
        {aho2007compilers}
\bibfield{author}{\bibinfo{person}{A.V. Aho}, \bibinfo{person}{M.S. Lam},
  \bibinfo{person}{R. Sethi}, {and} \bibinfo{person}{J.D. Ullman}.}
  \bibinfo{year}{2007}\natexlab{}.
\newblock \bibinfo{booktitle}{{\em Compilers: Principles, Techniques, and
  Tools}}.
\newblock \bibinfo{publisher}{ADDISON WESLEY Publishing Company Incorporated}.
\newblock
\showISBNx{9780321547989}
\showURL{%
\url{https://books.google.com.br/books?id=WomBPgAACAAJ}}


\bibitem[\protect\citeauthoryear{Albert and Barab\'asi}{Albert and
  Barab\'asi}{2002}]%
        {albertbarabasi}
\bibfield{author}{\bibinfo{person}{R\'eka Albert} {and}
  \bibinfo{person}{Albert-L\'aszl\'o Barab\'asi}.}
  \bibinfo{year}{2002}\natexlab{}.
\newblock \showarticletitle{Statistical mechanics of complex networks}.
\newblock \bibinfo{journal}{{\em Rev. Mod. Phys.\/}}  \bibinfo{volume}{74}
  (\bibinfo{date}{Jan} \bibinfo{year}{2002}), \bibinfo{pages}{47--97}.
\newblock
Issue 1.
\showDOI{%
\url{https://doi.org/10.1103/RevModPhys.74.47}}


\bibitem[\protect\citeauthoryear{Azimov and Grigorev}{Azimov and
  Grigorev}{2017}]%
        {azimov-grigorev2017matrix}
\bibfield{author}{\bibinfo{person}{Rustam Azimov} {and} \bibinfo{person}{Semyon
  Grigorev}.} \bibinfo{year}{2017}\natexlab{}.
\newblock \bibinfo{title}{{G}raph {P}arsing by {M}atrix {M}ultiplication}.
\newblock   (\bibinfo{year}{2017}).
\newblock
\showeprint{1707.01007}
\newblock
\shownote{arXiv:1707.01007v1.}


\bibitem[\protect\citeauthoryear{Grigorev and Ragozina}{Grigorev and
  Ragozina}{2016}]%
        {grigorev2016ll}
\bibfield{author}{\bibinfo{person}{Semyon Grigorev} {and}
  \bibinfo{person}{Anastasiya Ragozina}.} \bibinfo{year}{2016}\natexlab{}.
\newblock \showarticletitle{Context-Free Path Querying with Structural
  Representation of Result}.
\newblock \bibinfo{journal}{{\em arXiv preprint arXiv:1612.08872\/}}
  (\bibinfo{year}{2016}).
\newblock


\bibitem[\protect\citeauthoryear{Grune and Jacobs}{Grune and Jacobs}{2007}]%
        {grune2007parsing}
\bibfield{author}{\bibinfo{person}{D. Grune} {and} \bibinfo{person}{C.J.H.
  Jacobs}.} \bibinfo{year}{2007}\natexlab{}.
\newblock \bibinfo{booktitle}{{\em Parsing Techniques: A Practical Guide}}.
\newblock \bibinfo{publisher}{Springer New York}.
\newblock
\showISBNx{9780387689548}
\showLCCN{2007936901}
\showURL{%
\url{https://books.google.com.br/books?id=05xA\_d5dSwAC}}


\bibitem[\protect\citeauthoryear{Hellings}{Hellings}{2014}]%
        {Hellings14}
\bibfield{author}{\bibinfo{person}{Jelle Hellings}.}
  \bibinfo{year}{2014}\natexlab{}.
\newblock \showarticletitle{Conjunctive Context-Free Path Queries}. In
  \bibinfo{booktitle}{{\em Proc. 17th International Conference on Database
  Theory (ICDT), Athens, Greece, March 24-28, 2014}},
  \bibfield{editor}{\bibinfo{person}{Nicole Schweikardt},
  \bibinfo{person}{Vassilis Christophides}, {and} \bibinfo{person}{Vincent
  Leroy}} (Eds.). \bibinfo{publisher}{OpenProceedings.org},
  \bibinfo{pages}{119--130}.
\newblock
\showDOI{%
\url{https://doi.org/10.5441/002/icdt.2014.15}}


\bibitem[\protect\citeauthoryear{Hellings}{Hellings}{2015}]%
        {Hellings2015pathresults}
\bibfield{author}{\bibinfo{person}{Jelle Hellings}.}
  \bibinfo{year}{2015}\natexlab{}.
\newblock \showarticletitle{Path Results for Context-free Grammar Queries on
  Graphs}.
\newblock \bibinfo{journal}{{\em CoRR\/}}  \bibinfo{volume}{abs/1502.02242}
  (\bibinfo{year}{2015}).
\newblock


\bibitem[\protect\citeauthoryear{Kuijpers, Fletcher, Yakovets, and
  Lindaaker}{Kuijpers et~al\mbox{.}}{2019}]%
        {jochemkuijpers}
\bibfield{author}{\bibinfo{person}{Jochem Kuijpers}, \bibinfo{person}{George
  Fletcher}, \bibinfo{person}{Nikolay Yakovets}, {and} \bibinfo{person}{Tobias
  Lindaaker}.} \bibinfo{year}{2019}\natexlab{}.
\newblock \showarticletitle{An Experimental Study of Context-Free Path Query
  Evaluation Methods}. In \bibinfo{booktitle}{{\em Proceedings of the 31st
  International Conference on Scientific and Statistical Database Management}}.
  ACM, \bibinfo{pages}{121--132}.
\newblock


\bibitem[\protect\citeauthoryear{Medeiros, Musicante, and Costa}{Medeiros
  et~al\mbox{.}}{2019a}]%
        {MEDEIROS201975}
\bibfield{author}{\bibinfo{person}{Ciro~M. Medeiros},
  \bibinfo{person}{Martin~A. Musicante}, {and} \bibinfo{person}{Umberto~S.
  Costa}.} \bibinfo{year}{2019}\natexlab{a}.
\newblock \showarticletitle{LL-based query answering over RDF databases}.
\newblock \bibinfo{journal}{{\em Journal of Computer Languages\/}}
  \bibinfo{volume}{51} (\bibinfo{year}{2019}), \bibinfo{pages}{75 -- 87}.
\newblock
\showISSN{2590-1184}
\showDOI{%
\url{https://doi.org/10.1016/j.cola.2019.02.002}}


\bibitem[\protect\citeauthoryear{Medeiros, Musicante, and Costa}{Medeiros
  et~al\mbox{.}}{2019b}]%
        {comlan}
\bibfield{author}{\bibinfo{person}{Ciro~M. Medeiros},
  \bibinfo{person}{Martin~A. Musicante}, {and} \bibinfo{person}{Umberto~S.
  Costa}.} \bibinfo{year}{2019}\natexlab{b}.
\newblock \showarticletitle{LL-based query answering over RDF databases}.
\newblock \bibinfo{journal}{{\em Journal of Computer Languages\/}}
  \bibinfo{volume}{51} (\bibinfo{year}{2019}), \bibinfo{pages}{75 -- 87}.
\newblock
\showISSN{2590-1184}
\showDOI{%
\url{https://doi.org/10.1016/j.cola.2019.02.002}}


\bibitem[\protect\citeauthoryear{Mendelzon and Wood}{Mendelzon and
  Wood}{1989}]%
        {mendelzon1989regular-simple-paths}
\bibfield{author}{\bibinfo{person}{A.~O. Mendelzon} {and}
  \bibinfo{person}{P.~T. Wood}.} \bibinfo{year}{1989}\natexlab{}.
\newblock \showarticletitle{Finding Regular Simple Paths in Graph Databases}.
  In \bibinfo{booktitle}{{\em Proceedings of the 15th International Conference
  on Very Large Data Bases}} {\em (\bibinfo{series}{VLDB '89})}.
  \bibinfo{publisher}{Morgan Kaufmann Publishers Inc.}, \bibinfo{address}{San
  Francisco, CA, USA}, \bibinfo{pages}{185--193}.
\newblock
\showISBNx{1-55860-101-5}
\showURL{%
\url{http://dl.acm.org/citation.cfm?id=88830.88850}}


\bibitem[\protect\citeauthoryear{Mendelzon and Wood}{Mendelzon and
  Wood}{1995}]%
        {MendelzonW95}
\bibfield{author}{\bibinfo{person}{Alberto~O. Mendelzon} {and}
  \bibinfo{person}{Peter~T. Wood}.} \bibinfo{year}{1995}\natexlab{}.
\newblock \showarticletitle{Finding Regular Simple Paths in Graph Databases.}
\newblock \bibinfo{journal}{{\em SIAM J. Comput.\/}} \bibinfo{volume}{24},
  \bibinfo{number}{6} (\bibinfo{year}{1995}), \bibinfo{pages}{1235--1258}.
\newblock
\showURL{%
\url{http://dblp.uni-trier.de/db/journals/siamcomp/siamcomp24.html#MendelzonW95}}


\bibitem[\protect\citeauthoryear{P{\'e}rez, Arenas, and Gutierrez}{P{\'e}rez
  et~al\mbox{.}}{2010}]%
        {nsparql}
\bibfield{author}{\bibinfo{person}{Jorge P{\'e}rez}, \bibinfo{person}{Marcelo
  Arenas}, {and} \bibinfo{person}{Claudio Gutierrez}.}
  \bibinfo{year}{2010}\natexlab{}.
\newblock \showarticletitle{nSPARQL: A navigational language for \{RDF\}}.
\newblock \bibinfo{journal}{{\em Web Semantics: Science, Services and Agents on
  the World Wide Web\/}} \bibinfo{volume}{8}, \bibinfo{number}{4}
  (\bibinfo{year}{2010}), \bibinfo{pages}{255 -- 270}.
\newblock
\showISSN{1570-8268}
\showDOI{%
\url{https://doi.org/10.1016/j.websem.2010.01.002}}
\newblock
\shownote{Semantic Web Challenge 2009User Interaction in Semantic Web
  research.}


\bibitem[\protect\citeauthoryear{Santos, Costa, and Musicante}{Santos
  et~al\mbox{.}}{2018}]%
        {fred}
\bibfield{author}{\bibinfo{person}{Fred~C. Santos}, \bibinfo{person}{Umberto~S.
  Costa}, {and} \bibinfo{person}{Martin~A. Musicante}.}
  \bibinfo{year}{2018}\natexlab{}.
\newblock \showarticletitle{A Bottom-Up Algorithm for Answering Context-Free
  Path Queries in Graph Databases}. In \bibinfo{booktitle}{{\em Web
  Engineering}}, \bibfield{editor}{\bibinfo{person}{Tommi Mikkonen},
  \bibinfo{person}{Ralf Klamma}, {and} \bibinfo{person}{Juan Hern{\'a}ndez}}
  (Eds.). \bibinfo{publisher}{Springer International Publishing},
  \bibinfo{address}{Cham}, \bibinfo{pages}{225--233}.
\newblock
\showISBNx{978-3-319-91662-0}


\bibitem[\protect\citeauthoryear{Scott and Johnstone}{Scott and
  Johnstone}{2010}]%
        {gll}
\bibfield{author}{\bibinfo{person}{Elizabeth Scott} {and}
  \bibinfo{person}{Adrian Johnstone}.} \bibinfo{year}{2010}\natexlab{}.
\newblock \showarticletitle{GLL Parsing}.
\newblock \bibinfo{journal}{{\em Electronic Notes in Theoretical Computer
  Science\/}} \bibinfo{volume}{253}, \bibinfo{number}{7}
  (\bibinfo{year}{2010}), \bibinfo{pages}{177 -- 189}.
\newblock
\showISSN{1571-0661}
\showDOI{%
\url{https://doi.org/10.1016/j.entcs.2010.08.041}}
\newblock
\shownote{Proceedings of the Ninth Workshop on Language Descriptions Tools and
  Applications (LDTA 2009).}


\bibitem[\protect\citeauthoryear{Tomita}{Tomita}{1985}]%
        {tomita}
\bibfield{author}{\bibinfo{person}{Masaru Tomita}.}
  \bibinfo{year}{1985}\natexlab{}.
\newblock \bibinfo{booktitle}{{\em Efficient Parsing for Natural Language: A
  Fast Algorithm for Practical Systems}}.
\newblock \bibinfo{publisher}{Kluwer Academic Publishers},
  \bibinfo{address}{Norwell, MA, USA}.
\newblock
\showISBNx{0898382025}


\bibitem[\protect\citeauthoryear{Valiant}{Valiant}{1975}]%
        {Valiant75}
\bibfield{author}{\bibinfo{person}{Leslie~G. Valiant}.}
  \bibinfo{year}{1975}\natexlab{}.
\newblock \showarticletitle{General Context-Free Recognition in Less than Cubic
  Time}.
\newblock \bibinfo{journal}{{\em J. Comput. Syst. Sci.\/}}
  \bibinfo{volume}{10}, \bibinfo{number}{2} (\bibinfo{year}{1975}),
  \bibinfo{pages}{308--315}.
\newblock
\showDOI{%
\url{https://doi.org/10.1016/S0022-0000(75)80046-8}}


\bibitem[\protect\citeauthoryear{W3C}{W3C}{2012}]%
        {w3c2012sparql-query-lang}
\bibfield{author}{\bibinfo{person}{W3C}.} \bibinfo{year}{2012}\natexlab{}.
\newblock \bibinfo{title}{{SPARQL} 1.1 Query Language}.
\newblock   (\bibinfo{year}{2012}).
\newblock
\showURL{%
\url{https://www.w3.org/TR/2012/PR-sparql11-query-20121108/}}


\bibitem[\protect\citeauthoryear{Zhang, Feng, Wang, Rao, and Wu}{Zhang
  et~al\mbox{.}}{2016}]%
        {zhang2016}
\bibfield{author}{\bibinfo{person}{Xiaowang Zhang}, \bibinfo{person}{Zhiyong
  Feng}, \bibinfo{person}{Xin Wang}, \bibinfo{person}{Guozheng Rao}, {and}
  \bibinfo{person}{Wenrui Wu}.} \bibinfo{year}{2016}\natexlab{}.
\newblock \showarticletitle{Context-Free Path Queries on {RDF} Graphs}. In
  \bibinfo{booktitle}{{\em International Semantic Web Conference {(1)}}} {\em
  (\bibinfo{series}{Lecture Notes in Computer Science})},
  Vol.~\bibinfo{volume}{9981}. \bibinfo{pages}{632--648}.
\newblock


\end{thebibliography}

%\section*{Appendices\ciro{NOT ALLOWED!!!}
%\appendix

%\section{Complete Version of Example~\ref{ex:example2}}

%\section{Proofs}
%\label{sec:proofs}

\end{document}